\documentclass[11pt]{article} 

\usepackage{microtype}
\usepackage{dsfont}
\usepackage{graphicx}
\usepackage{subfigure}
\usepackage{booktabs} 
\usepackage[shortlabels]{enumitem}
\usepackage{algorithm}
\usepackage{algorithmic}

\usepackage{hyperref}

\usepackage{amsmath}
\usepackage{amssymb}
\usepackage{mathtools}
\usepackage{amsthm}
\usepackage{xspace}
\usepackage[toc,page]{appendix}
\usepackage[capitalize,noabbrev]{cleveref}
\usepackage[numbers]{natbib}
\usepackage[utf8]{inputenc}
\usepackage{afterpage}

\usepackage{epstopdf}
\usepackage{underlin}
\usepackage{fancyhdr}

\usepackage{pdfpages}

\usepackage{color}
\usepackage{pgfplots}
\usepackage{fullpage}

\theoremstyle{plain}
\newtheorem{theorem}{Theorem}[section]

\newtheorem{proposition}[theorem]{Proposition}
\newtheorem{lemma}[theorem]{Lemma}

\theoremstyle{definition}
\newtheorem{definition}[theorem]{Definition}
\newtheorem{assum}[theorem]{Assumption}
\theoremstyle{remark}
\newtheorem{remark}[theorem]{Remark}

\newcommand{\R}{\mathbb{R}}
\newcommand{\E}{\mathbb{E}}

\newcommand{\X}{\mathbf{X}}
\newcommand{\Y}{\mathbf{Y}}
\newcommand{\Z}{\mathbf{Z}}

\newcommand{\z}{\mathbf{z}}
\newcommand{\ab}{\mathbf{a}}

\newcommand{\Rn}{\mathbb{R}^n}

\newcommand{\T}{\top}
\newcommand{\dd}{\mathrm{d}}
\newcommand{\I}{\mathcal{I}}
\newcommand{\id}{\mathrm{I}}

\newcommand{\abs}[1]{\ensuremath{\left\vert #1\right\vert}}
\newcommand{\norm}[1]{\ensuremath{\left\| #1\right\|}}
\newcommand{\setp}[1]{\ensuremath{\left\{ #1\right\}}}

\DeclarePairedDelimiter\ceil{\lceil}{\rceil}

\title{Provably Private Distributed Averaging Consensus: \\
	An Information-Theoretic Approach}

\author{
	{\normalfont Mohammad Fereydounian}\thanks{Department of Electrical and Systems Engineering, University of Pennsylvania, Philadelphia, PA, USA. Emails: \ttfamily\bfseries\{mferey, hassani\}@seas.upenn.edu} 	
	\and {Aryan Mokhtari}\thanks{Department of Electrical and Computer Engineering, University of Texas at Austin, Austin, TX, USA. Email: \ttfamily\bfseries mokhtari@austin.utexas.edu}
	\and {Ramtin Pedarsani} \thanks{Department of Electrical and Computer Engineering, University of California, Santa Barbara, CA, USA. Email: \ttfamily\bfseries ramtin@ece.ucsb.edu} 
	\and {Hamed Hassani}\footnotemark[1]
}

\date{}
\begin{document}

	\maketitle

\begin{abstract}
 In this work, we focus on solving a decentralized consensus problem in a private manner. Specifically, we consider a setting in which a group of nodes, connected through a network, aim at computing the mean of their local values without revealing those values to each other. The distributed consensus problem is a classic problem that has been extensively studied and its convergence characteristics are well-known. Alas, state-of-the-art consensus methods build on the idea of exchanging local information with neighboring nodes which leaks information about the users' local values. We propose an algorithmic framework that is capable of achieving the convergence limit and rate of classic consensus algorithms while keeping the users’ local values private. The key idea of our proposed method is to carefully design noisy messages that are passed from each node to its neighbors such that the consensus algorithm still converges precisely to the average of local values, while a minimum amount of information about local values is leaked. We formalize this by precisely characterizing the mutual information between the private message of a node and all the messages that another adversary collects over time. We prove that our method is capable of preserving users’ privacy for any network without a so-called {\it generalized leaf}, and formalize the trade-off between privacy and convergence time. Unlike many private algorithms, any desired accuracy is achievable by our method, and the required level of privacy only affects the convergence time.
\end{abstract}

\section{Introduction}

In this paper, we focus on a classic distributed computing problem in which a group of connected agents aim to find the average of their local values, also known as the \textit{consensus problem}. Due to  applications of the consensus problem in many domains, such as sensor fusion \cite{olfati2005consensus,xiao2005scheme,he2015multiperiod}, distributed energy management \cite{zhao2016consensus}, Internet of Things \cite{chen2017narrowband}, % time synchronization \cite{schenato2011average,he2013time}, 
and large-scale machine learning and federated learning \cite{bekkerman2011scaling, mcmahan2017communication}, it has been widely studied in the literature \cite{xiao2004fast,olfati2004consensus,jadbabaie2003coordination}.

A common feature in most classic consensus algorithms is the requirement for sharing the local values with neighboring nodes. However, this approach is, indeed, problematic in settings that nodes are not willing to share their exact local value (opinion or belief) due to \textit{privacy} concerns \cite{mo2016privacy}. An example would be in the case where members of a social network aim to compute their common opinion on a subject, but want to keep their personal opinions secret \cite{degroot1974reaching}. This issue has motivated a new line of research which focuses on the possibility of achieving consensus in a fully decentralized manner without disclosing the initial nodes' value.

To provide privacy in the consensus problem, the first step  is to define a measure that quantifies  it. Differential privacy \cite{cortes2016differential,dwork2014algorithmic} is the most studied measure of privacy for an algorithm running over a dataset. This notion measures the privacy based on the statistical dependency of an algorithm's output to the perturbation of a single element of the input dataset. Noting that the classic consensus method is a deterministic procedure, the most common approach for a private consensus algorithm in the literature is that each node perturbs its signals by some form of noise such that the resulting stochastic algorithm is differentially private \cite{huang2012differentially,nozari2017differentially,feldman2014sample,chaudhuri2013stability,bassily2014private,kamath2019privately,heikkila2017differentially,he2018privacy,mo2016privacy,le2013differentially}. However, these approaches suffer from some drawbacks. Adding perturbing noises  affects the convergence properties compared to non-private consensus algorithms in two ways: (i) It leads to a non-exact limit, and (ii) it compromises the convergence rate, i.e., it leads to slower convergence rates, e.g., slower than the rates achieved by \cite{xiao2004fast}. While some methods address the non-exact limit by adding zero-sum correlated noise \cite{he2018privacy,mo2016privacy}, no prior study has addressed both (i) and (ii) simultaneously. This issue persists for any iterative method that  perturbs in some way the communication messages at each iteration either randomly or in a deterministic manner. This includes works that use methods like output masking \cite{altafini2020system}, state decomposition \cite{wang2019privacy}, splitting messages into segments and adding noise to each segment \cite{heikkila2017differentially}, observability methods \cite{kia2015dynamic,pequito2014design}, edge-based perturbation methods \cite{gupta2018information,xiong2021privacy}, and Homomorphic encryption-based methods \cite{hadjicostis2018privary,yin2019accurate,ruan2019secure}.

Indeed, differential-privacy is a powerful measure for indicating the level of privacy of a learning algorithm running over a dataset. This is done by quantifying the statistical dependency of the algorithm's output to a perturbation of a single input data point. This approach is best applicable when a full list of data points is accessible by the adversary.  Alas, in a consensus setting, the data available at one user, i.e., the messages it collects from its neighbors over time, is different from the data available at other users. By injecting extra noise to all communication messages in a consensus algorithm, the differential-privacy approaches mask the private values against an adversary that can eavesdrop all messages. The cost to defeat this adversary appears on convergence side as discussed above. However, in a setting that the adversary can gain access to only a single node's information, such approaches provide more than necessary privacy with an unwanted cost. Hence, to handle the consensus setting with such type of adversary and asymmetric data distribution, differential privacy must be replaced with a more compatible privacy measure. 

Motivated by this point, we provide a novel information-theoretic scheme to measure and analyze privacy leakage in consensus problems. We further provide simple noise-aggregation methods to reach the exact consensus in a private manner while achieving the fastest rate that a non-private consensus method can achieve. These advantages are based on the fact that our proposed \textit{provably private averaging consensus} (PPAC) method only modifies the initial values of the consensus dynamic and leaves the consensus procedure intact. Moreover, we use a probabilistic model where the private values are random variables, which makes our method fundamentally different from most of the prior work. In this probabilistic framework, we introduce a new measure of privacy leakage based on the mutual information between the initial value of a node and the set of messages that is available at another node. In this way, we can capture how private the initial value of a node is with respect to any other node in the network. Most of our theoretical results hold regardless of the shape of the distribution of noises and private values. Next, we state our contributions:
\vspace{-1.5mm}
\begin{enumerate}
	\item Given a network structure and a consensus matrix, we propose an algorithm for reaching consensus among the nodes,  while keeping the initial values private throughout the algorithm and preserving the fastest rate that a non-private consensus method can achieve.
	\vspace{-1.5mm}
	\item We formalize a continuous notion based on the mutual information function that measures the privacy leakage of any (victim) node at another (adversary) node that can be (optimally) tuned by our adjustable noise parameters.  
	\vspace{-1.5mm}
	\item We provide a precise characterization of the information flow of a private value over the network over time. Having this, even in ill-shaped structures that leak privacy, we can determine the amount of information that an adversary receives as a function of time and compute the  waiting time until recovering a private value.
\end{enumerate}
In summary, this paper has two main messages: (i) Privacy from a node's perspective can be measured and efficiently analyzed based on the concept of mutual information. (ii) Splitting private messages into fragments prior to running the classical consensus method maintains the convergence guarantees and ensures the privacy for most structures. No further additive noise is required.

\noindent\textbf{Notation.}
Column vectors are denoted by small letters in bold font, $\mathbf{a},\mathbf{b}$, while matrices are denoted by capital non-bold letters, $A,B$, and scalars by small letters $a,b$. Moreover, $\{a_i\}_{i\in I}$ is considered as an ordered sequence of mathematical objects $a_i$, $i\in I$. Concatenation of ordered sequences is denoted by Cartesian product and when the number of elements are finite, these ordered sequences are considered as column vectors. For scalars $a_i$, the following example illustrates these notations: 
\begin{equation}\label{ac1}
	\setp{a_i}_{i\in\{1,2\}}\times (a_3,a_4) \times \setp{a_5} = \left[a_1, a_2, a_3, a_4, a_5\right]^{\top}.
\end{equation}
The right-hand side of \eqref{ac1} indicates a column vector with elements $a_1, a_2, a_3, a_4, a_5$. When $a_i\,$s in \eqref{ac1} are replaced with row vectors, \eqref{ac1} refers to the matrix with rows $a_1, a_2, a_3, a_4, a_5$. Moreover, note that $\setp{a_i}_{i\in\{1,2\}}$ can be interpreted as either $[a_1,a_2]^{\top}$ or $[a_2,a_1]^{\top}$. In this paper, whenever the choice of order is not specified, the argument holds regardless of the choice. The identity matrix of size $k$ is $I_{k}$ and all ones (column) vector of length $k$ is $\mathds{1}_k$. We use both $A_{ij}$ and $[A]_{ij}$ to denote the $ij$-th element of the matrix $A$. Also, we use $[A]_{i*}$ and $[A]_{*j}$ to represent the $i$-th row of $A$ as a row vector and the $j$-th column of $A$ as a column vector. Similarly, $[\mathbf{v}]_i$ refers to the 
$i$-th coordinate of the vector $\mathbf{v}$. Further, $A \geq 0$ means $A_{ij}\geq 0$ for all $i,j$. Further, The notation $\operatorname{diag}(a_1,\ldots,a_n)$ represents an $n\times n$ diagonal matrix with $a_i$ as its $i$-th diagonal element. If $A=\setp{\ab_1,\ldots,\ab_k}\subset \Rn$, then $\operatorname{span}(A) = \{\sum_{i=1}^{k}t_i\ab_i \mid t_1,\ldots,t_k\in\R\}$. Further, $[n]=\{1,\ldots,n\}$ and $A\setminus B$ represents set difference for sets $A$ and $B$. Finally, $|A|$ denotes the size of set $A$ and for (real or complex) scalar $a$, $|a|$ denotes the absolute value of $a$.

\section{Preliminaries}\label{pre}

In this section, we recap some basic concepts that we require for presenting  our framework.

\noindent \textbf{Graph Theory.}
By $G=(V,E)$, we denote a simple graph where $V$ is the set of nodes and $E$ is the set of (undirected) edges. The distinct nodes $i$ and $j$ are called neighbors if $\{i,j\}\in E$. Further, the neighborhood of node $i$, denoted by $N_{i}$, is the set of all neighbors of node $i$. The degree of node $i$ is $\deg_{G}(i)=|N_i|$. A walk of length $k$ between $i$ and $j$, is the sequence of nodes $(\ell_0,\ell_{1},\ldots,\ell_{k})$, where $\ell_0 = i$, $\ell_k=j$, and all consecutive nodes are neighbors. A graph is called connected if for every node pair, there is a walk between them.
The minimum $k$ for which a walk of length $k$ exists between $i$ and $j$ is called the distance between $i$ and $j$ and denoted by $d_G(i,j)$. The eccentricity of node $i$ is the largest distance we can get from node $i$, i.e., $\operatorname{ecc}_G(i) = \max_{j\in V}d_G(i,j)$ and the radius of $G$ is defined as $r(G) = \min_{i\in V} \operatorname{ecc}_G(i)$. The adjacency matrix of a graph $G$ with $n$ nodes is an $n\times n$ matrix denoted by $A_G$, where $\left[A_G\right]_{ij}$ is $1$ if $\{i,j\}\in E$ and $0$ otherwise. 

\noindent \textbf{Matrix Theory.}
For matrices $A$ and $B$, we define $A \sim B$ if for every $i\neq j$, $A_{ij} =0$ if and only if $B_{ij} = 0$. 
Spectral radius of $A\in \R^{n \times n}$ with eigenvalues $\lambda_1,\ldots,\lambda_n$ is  $\rho(A) = \max_{i\in [n]}\{|\lambda_i|\}$.
For $A\in \R^{n \times n}$, its minimal polynomial $\mu_A$ is defined as the unique monic polynomial with minimum degree such that $\mu_A(A)=0$. By Cayley-Hamilton theorem, $\deg(\mu_A)\leq n$ (see \cite{roman2013advanced}). 

Given an infinite sequence of vectors $\{\ab_i\}_{i=0}^{\infty}$, where $\ab_i\in\Rn$ and $\ab_0\neq \mathbf{0}$, we consider a basis pursuit procedure which constructs a finite subset $B$ of the elements of the sequence $\setp{\ab_i}_{i=0}^{\infty}$ such that the elements of $B$ are linearly independent and $\ab_i\in\operatorname{span}(B)$ for all $i$ and in this sense, it is a called a {\it basis} for $\{\ab_i\}_{i=0}^{\infty}$ and is denoted by $B = \mathcal{B}(\{\ab_i\}_{i=0}^{\infty})$. The procedure is as follows: Add $\ab_0$ to $B$. For $i>0$, if $\ab_{i}\notin \operatorname{span}(B)$, then add  $\ab_{i}$ to $B$ and go to the next iteration. We formally prove in the supplementary material that this procedure results in a basis for $\{\ab_i\}_{i=0}^{\infty}$.

\noindent \textbf{Information Theory.}
Consider continuous random variables $X$ and $Y$ which are defined over spaces $\mathcal{X}$ and $\mathcal{Y}$ with probability density functions (PDFs) $f_X$ and $f_{Y}$, respectively. Moreover, let $f_{XY}(x,y)$ denote joint PDF of $X$ and $Y$, then their mutual information $\mathcal{I}(X;Y)$ is defined as
\begin{equation}\label{mut}
	\mathcal{I}(X;Y) = \int_{\mathcal{X}\times \mathcal{Y}}f_{XY}(x,y)\log\left(\frac{f_{XY}(x,y)}{f_X(x)f_Y(y)}\right)dxdy.
\end{equation}

\noindent \textbf{Classical Averaging Consensus.}
Consider a network represented by a simple connected graph $G=(V,E)$ over a set of nodes $V=\setp{1,\ldots,n}$ with $m=|E|$. We assume only adjacent nodes can exchange information with each other. Moreover, suppose $u_i\in \R$ is the initial value of user $i$. The main goal in the consensus problem is to reach a state that all nodes in the network learn the average of initial vectors, i.e., $u^*=\frac{1}{n}\sum_{i=1}^n u_i$. 
Letting $v_i(t)$ be the value of node $i$ at iteration $t$ which is initially set to $v_i(0)=u_i$, one can consider a linear update given by
\begin{equation}\label{update}
	v_i(t+1) = W_{ii}\, v_i(t) + \sum_{j\in N_i}\,W_{ij} v_j(t).
\end{equation} 
Having this, the averaging consensus aims to have $\lim_{t \to \infty } v_i(t) = u^*$ for all $i\in[n]$. {\it The consensus matrix} $W\in\R^{n\times n}$ is formed by considering $W_{ij}$ as its $ij$-th element. Letting $\mathbf{v}(t) = [v_1(t),\ldots,v_n(t)]^{\T}$, one can write $\mathbf{v}(t)=W^t\mathbf{v}(0)$ seeking to have $\lim_{t \to \infty } \mathbf{v}(t) = \mathbf{u}^*$, where $\mathbf{u}^*= u^*\mathds{1}_n$. This is equivalent to
\begin{align}\label{cons conv}
	\lim_{t \to \infty } W^t = \frac{1}{n}\mathds{1}_n\mathds{1}_n^{\T}.
\end{align} 
It is known that \cref{cons conv} holds if and only if (i) $W\mathds{1}_n = W^{\T}\mathds{1}_n = \mathds{1}_n$ and (ii) $\rho( W - {1}/{n}\mathds{1}_n\mathds{1}_n^{\T}) <1$, 
where $\rho(\cdot)$ denotes the spectral radius of a matrix; see \cite{xiao2004fast} for more details. Further, the convergence rate of \eqref{cons conv} can be computed as follows:
\begin{equation}\label{conv rate}
	\sup _{\mathbf{v}(0) \neq \mathbf{u}^*} \lim _{t \rightarrow \infty}\left(\frac{\|\mathbf{v}(t)-\mathbf{u}^*\|_{2}}{\|\mathbf{v}(0)-\mathbf{u}^*\|_{2}}\right)^{\frac 1t} = \rho\left( W - \frac{1}{n}\mathds{1}_n\mathds{1}_n^{\T}\right).
\end{equation} 

\vspace{-2mm}
\section{Problem Setup}\label{sec: setup}
We consider the averaging dynamic described in \eqref{update} under the constraint that the original messages $u_i$ must be kept private from other nodes. Satisfying this requirement rules out the naive initialization of $v_i(0)=u_i$. Therefore, we seek new initializations $v_i(0)$ to fulfill the privacy and convergence requirements while following the consensus dynamic \eqref{update}.
Moreover, we assume that the initial values $u_i$ are independent scalar-valued continuous random variables, but a similar argument is applicable to a discrete case. Moreover, for a vector-valued $u_i$ with independent coordinates, our results can be applied to each coordinate.

\noindent \textbf{Adversary model.}
We assume all nodes are honest but curious, meaning that they follow the protocols and communication principles of the network, but they might be willing to recover the initial value of other nodes in the network. This is equivalent to the situation in which an adversary obtains access to the messages received by a single node. 
Moreover, we assume that nodes (adversaries) do not collude in recovering private messages and the underlying graph and the full consensus matrix are accessible by all nodes (adversaries).

\noindent \textbf{Privacy model.}
Consider arbitrary distinct nodes $i$ and $j$ and suppose node $i$ is curious about node $j$'s private value $u_j$. Node $i$ only has access to the $u_i$ and all messages it sends to or receives from its neighbors over time. Denote the concatenation of all these by $\mathfrak{D}_i$. The privacy fails if node $i$ can deterministically recover $u_j$ from $\mathfrak{D}_i$. Otherwise, it is important to know how close node $i$ can statistically estimate $u_j$ using $\mathfrak{D}_i$.
This can be evaluated by finding how much information $\mathfrak{D}_i$ reveals about $u_j$ which can be formalized by the mutual information between the joint distribution of the elements of $\mathfrak{D}_i$ and $u_j$, denoted by $\mathcal{I}(\mathfrak{D}_i\,;u_j)\in [0,\infty]$. Achieving zero mutual information is impossible since the consensus goal $u^*$ itself reveals some information about $u_j$. Upcoming sections reveal how tuning the model's parameters push $\mathcal{I}(\mathfrak{D}_i\,;u_j)$  toward~$0$ as much as possible.

\vspace{-1mm}
\section{Private Consensus}\label{PC}
\vspace{-1mm}
In this section, we first explain our method and then introduce a privacy leakage measure in \cref{PLM}. As mentioned earlier, the regular averaging consensus method 
reveals the private messages at $t=0$ due to the initialization $v_i(0)=u_i$. We propose a new initialization for $v_i(0)$ that preserves privacy while leaving most convergence properties intact. The main idea behind the proposed method is the following: If we modify the initial vectors $v_i(0)$ such that their sum preserves the sum of the right local values $u_i$, i.e., $\sum_{i=1}^n v_i(0) = \sum_{i=1}^n u_i$, then by following the averaging consensus dynamic all nodes converge to the optimal value $u^*=\frac{1}{n}\sum_{i=1}^n u_i$. Now our goal is to select these values so that they do not reveal the original values $\{u_i\}_{i\in [n]}$ or minimize the amount of information that is revealed after a specific number of communication rounds.

To do so, we proceed as follows. In the first round, which we call the preparation phase, unlike the traditional consensus approach, each node $i$ does not send the same signal to all its neighbors. Instead, it splits its private message $u_i$ into multiple pieces and distributes it among its neighbors. All but one of these pieces are pure noise terms, independent from $u_i$. 
More precisely, in the preparation phase, each node $i$ splits its private value $u_i$ into fragments $\Gamma_{ij}$, where $j\in N_i$. Hence, the number of fragments is equal to the number of neighbors of node $i$. These elements are selected such that all of them except one are pure noise and their sum recovers the original signal, i.e.,
\begin{equation}\label{split}
	u_i = \sum_{j\in N_i} \Gamma_{ij}.
\end{equation}
Formally, to create such signals, node $i$ arbitrarily chooses one of its neighbors, which we denote by $m_i$. For $j\neq m_i$, node $i$ sets $\Gamma_{ij}$ to be a continuous random variable, independent from all other randomness sources. Then, node $i$ sets $\Gamma_{i\, m_i}$ such that \eqref{split} holds, i.e., it sets
\begin{align}
	\Gamma_{i\,m_i} = u_i-\sum_{j\in N_i\setminus \{m_i\}}\Gamma_{ij}.
\end{align}
The following indexing set $S$ distinguishes all pairs $(i,j)$ such that $\Gamma_{ij}$ is a pure noise term. 
\begin{equation}\label{index-set}
	S = \setp{(i,j)\,\, \mid i\in[n],\, j\in N_i\setminus \setp{m_i}}.
\end{equation}
Note that $|S|=2m-n$, where $m=|E|$. Randomness sources in this model consist of $n$ private values and $2m-n$ pure noisy messages. Next, we state the independence requirement.
\begin{assum}\label{as: 1}
	The concatenation of all $2m$ randomness sources in Algorithm~\ref{alg}, i.e., $\{u_i\}_{i\in [n]}\times \setp{\Gamma_{ij}}_{(i,j)\in S}$, where $S$ is defined as \eqref{index-set} consists of independent elements.
\end{assum}  
The choice of distribution for randomness sources is arbitrary among continuous random variables with a valid PDF unless otherwise is specified. However, it is beneficial to fix a notation for their mean and variances as follows.
\begin{definition}\label{def: parameters}
	For $i\in [n]$ and $(k,\ell)\in S$, let $\mu_i$ and $\sigma_i^2$ be the mean and variance of $u_i$ and $\mu_{k\ell}$ and $\sigma_{k\ell}^2$ be the mean and variance of $\Gamma_{k\ell}$, respectively. Moreover, let $\mu_{\max}$ be the maximum of the absolute value of all $\mu_i$ and $\mu_{k\ell}$, and define $\sigma_{\max}^2$ in a similar manner for variances.
\end{definition}
\begin{algorithm}[t!]
	\caption{Provably Private Averaging Consensus (PPAC)}
	{\bf Preparation-Phase:} 
	\begin{algorithmic}[1]
		\FOR { $i\in [n]$}
		\STATE Node $i$ picks $m_i\in N_i$ arbitrarily.
		\STATE For all $j\in N_i\setminus \setp{m_i}$: node~$i$ generates noise $\Gamma_{ij}$ and sends it to node~$j$ 
		\STATE Node $i$ computes $\Gamma_{i\,m_i}=u_i-\sum_{j\in N_i\setminus \{m_i\}}\Gamma_{ij}$ and sends it to node $m_i$.
		\ENDFOR
		\FOR { $i \in [n]$}
		\STATE Node $i$ computes: $v_{i}{(0)} = \sum_{k\in N_i}\Gamma_{ki}$. \label{e}
		\ENDFOR
	\end{algorithmic}
	{\bf Consensus-Phase:} 
	\begin{algorithmic}[1]
		\STATE Nodes follow the consensus dynamic in \eqref{update} with initial values $\{v_{i}{(0)}\}_{i\in [n]}$. 
	\end{algorithmic}
	\label{alg}
\end{algorithm}
Once the preparation phase is done and the messages $\Gamma_{ij}$ are communicated, every node computes its consensus initialization by summing up the messages that it  has received in the preparation phase, that is, $v_{i}{(0)} = \sum_{k\in N_i}\Gamma_{ki}$. After that,  all nodes follow the consensus dynamic described in \eqref{update}. The steps of our proposed method (PPAC) are summarized in Algorithm~\ref{alg}.

Note that PPAC preserves the sum property after the preparation phase, i.e., $\sum_{i=1}^n v_i(0) = \sum_{i=1}^n u_i$. To achieve the exact limit, it is required to have $W\mathds{1}_n = W^{\T}\mathds{1}_n = \mathds{1}_n$. This implies that the sum property is preserved during the consensus phase too, i.e.,  $\sum_{i=1}^n v_i(t) =\sum_{i=1}^n v_i(0)= \sum_{i=1}^n u_i$.

\subsection{Privacy Leakage Measure}\label{PLM}
We first provide a mathematical definition to formalize the notion of privacy leakage that was earlier described in Section~\ref{sec: setup}. Then, we rigorously analyze this notion.

We start with mathematically formalizing the concatenation of all data that a node $i$ collects, which was earlier denoted by $\mathfrak{D}_i$ in \cref{sec: setup}.  
The first data available at node $i$ is its own private value $u_i$. Next, in the preparation phase, it generates $\{\Gamma_{i\ell}\}_{\ell \in N_i\setminus \{m_i\}}$. Note that $\Gamma_{i\,m_i}$ is a redundant data, since it is simply the subtraction of other pure noises from $u_i$. In the preparation phase, node $i$ receives $\{\Gamma_{\ell i}\}_{\ell\in N_i}$ from its neighbors. Moreover, when the consensus procedure starts, at time $t \geq 0$, node $\ell\in N_i$ transmits $v_{\ell}(t)$ to node $i$ so that it can compute its new update for the next iteration. 
Hence, up to time $t$, node $i$ has received values $v_{\ell}(\tau)$ for all $\ell\in N_i$ and $\tau \in \{0,\ldots,t\}$. Hence, the concatenation of node $i$'s observed data up to time $t$ can be written as
\begin{align}\label{eq11}
	\mathfrak{D}_i(t) = \setp{u_i}\times\setp{\Gamma_{i\ell}}_{\ell \in N_i\setminus \setp{m_i}}\times \setp{\Gamma_{\ell i}}_{\ell\in N_i}\times \setp{v_{\ell}(\tau)}_{\ell\in N_i, 0\leq \tau\leq t}.
\end{align}
The notation $\mathfrak{D}_i(\infty)$ is also applicable and represents all the data that node $i$ can gather if the consensus runs forever. To measure the privacy leakage of $u_j$ at node $i$ up to time $t$, we consider the mutual information between $\mathfrak{D}_i(t)$ and $u_j$. This is formalized in the following definition.
\begin{definition}[Privacy Leakage Measure]\label{def: pr}
Considering the definition of $\mathfrak{D}_i(t)$ in \eqref{eq11}, the privacy leakage of node $j$ from the perspective of node $i$ up to time $t$ is defined as 
	\begin{equation}\label{def_Inf}
		\pi_{i}^{(j)}(t) := \mathcal{I}\left(\mathfrak{D}_i(t)\,;\,u_j\right) \in  [0,\infty).
	\end{equation}  
	\vspace{-7mm}
\end{definition}
Note that smaller $\pi_{i}^{(j)}(t)$ means $u_j$ is more private at node $i$ as less information is revealed about it.

\section{Main Results}
Based on the private consensus method in \cref{PC}, our main results can be stated informally in \cref{informal}. Formal statements regarding privacy and convergence are provided in \cref{sec:me pr} and \cref{sec: conv}, respectively.

\begin{theorem}[informal]\label{informal}
	Running \cref{alg} (PPAC) over a network and assuming all randomness sources are independent, we can conclude the following results:
	\begin{enumerate}[(i)]
		\item \label{flow} For each node $i$, there are only finitely many iterations $t_1,\ldots,t_k$ that the information of node $i$ about some private values strictly increases (\cref{theorem:1}). For the latest of such iterations, i.e., $t_k$, we find an upper-bound $t_k\leq n-1$, where $n$ is the number of all nodes (Proposition~\ref{lem:3}), and a lower-bound  $\operatorname{ecc}_G(i)-2 \leq t_k$, where $\operatorname{ecc}_G(i)$ is the eccentricity of node $i$ (Proposition~\ref{lem: 4}). 
		\vspace{-1mm}
		\item If the network contains a specific sub-structure called a generalized leaf, defined in Definition~\ref{gene}, some nodes can deterministically recover the private values of some other nodes in the first consensus iteration, i.e., privacy fails.
		\vspace{-1mm}
		\item If the network does not contain a generalised leaf, no node can fully recover a private value (Lemma~\ref{lemma:case2} and \cref{theorem:2}). In this case, $\pi_{i}^{(j)}$, defined in Definition~\ref{def: pr}, measures the amount of leakage of $u_j$ to node $i$, statistically. We obtain a closed form  $\pi_{i}^{(j)}(t)=1/2\log(1+\sigma_j^2\,\ab^{\T}\Sigma^{-1} \ab)$ for some vector $\ab$ and matrix $\Sigma$ (\cref{lem:2}). In \cref{sec:me pr}, we discuss the construction of $\ab$ and $\Sigma$ through several steps. In this construction, the quantities $\sigma_{k\ell}^2$, the variances of noisy fragments $\Gamma_{k\ell},\,(k,\ell)\in S$, appear as linear terms in the elements of $\Sigma$. This shows that $\pi_{i}^{(j)}$ is a decreasing function in terms of $\sigma_{k\ell}^2$. 
		\vspace{-1mm}
		\item PPAC leaves the convergence limit and rate of the classical averaging consensus intact. Moreover, the convergence time to achieve an $\epsilon$-accurate solution by 
		PPAC scales as $\mathcal{O}(\log(1/\epsilon))$ and it also increases proportional to the logarithm of noise parameters (\cref{conv theorem}). This shows that tuning $\sigma_{k\ell}^2$ adjusts the trade-off between privacy and convergence time.
	\end{enumerate}
\end{theorem}

Next, we discuss the novel techniques used to achieve the above results. The baseline in obtaining our results is transforming the information-/graph-theoretical formulation of the problem into a linear-algebraic form by writing the data collection at each node as a matrix-vector decomposition $\mathfrak{D}_i(t)=\mathfrak{R}_i(t)\mathbf{g}$, where $\mathbf{g}$ is the vector consisting of the  $2m$ randomness sources of the model. As we show later, the possibility of leakage and properties of information flow over the graph can be translated into constraints on $\mathfrak{R}_i(t)$. 

 While the size of $\mathbf{g}$ is fixed, each new iteration adds $\deg(i)$ rows to $\mathfrak{D}_i(t)$ and $\mathfrak{R}_i(t)$. The  decomposition $\mathfrak{D}_i(t)=\mathfrak{R}_i(t)\mathbf{g}$  reveals that, after sufficiently many iterations, all new rows are linear combinations of previous ones. This leads to the fact that finitely many iterations $t_1,\ldots,t_k$ have new information. Moreover, we constructively find $t_1,\ldots,t_k$ by running a basis-pursuit procedure on the rows of $\mathfrak{R}_i(t)$. As a next step, we write each element $\mathfrak{R}_i(t)$ as a function of the elements of $W^t$. Having this, we then use Cayley-Hamilton theorem (see \cite{roman2013advanced}) on matrix powers to show that $t_k\leq n-1$. The aforementioned decomposition also paves the way to translate the waiting time for node $i$ to receive the first message containing non-redundant information about $u_j$ in terms of the distance $d_G(i,m_j)$, which leads to $\operatorname{ecc}_G(i)-2 \leq t_k$.  

Considering the data collection $\mathfrak{D}_i(t)$ on only non-redundant data points corresponding to $t_1,\ldots,t_r$ for $r\leq k$ gives $\mathfrak{D}_i^r$ and accordingly $\mathfrak{R}_i^r$. The next tool to obtain the results of \cref{informal} is the fact that we translated the ability of deterministic recovery of $u_j$ by node $i$ into the rank-deficiency of $\mathfrak{R}_i^r$ when its $j$-th column is removed. Having this, the core technique in obtaining the main results of this paper is transforming the aforementioned rank-deficiency condition into the existence of a specific substructure in the graph called a \emph{generalized leaf}. A generalized leaf with head $j$ and tail $i$ consists of $i$ being connected to all degree $2$ neighbors of $j$ (and possibly $j$ itself) which is illustrated in \cref{leaf}. While it is straightforward to see why in a generalized leaf with head $j$ and tail $i$, node $i$ can fully recover $u_j$, it is far more challenging to prove that this is indeed the only possible scenario for a full recovery of some private value. We prove this fact based on a refined analysis (see \cref{main-proof}) that shows, assuming there are no  generalized leaves, the $j$-th column of $\mathfrak{R}_i^r$ can be written in terms of other columns and thus removing the $j$-th column does not decrease the rank of $\mathfrak{R}_i^r$.  
\begin{figure}[t!]
	\centering
	\includegraphics[width=0.7\textwidth]{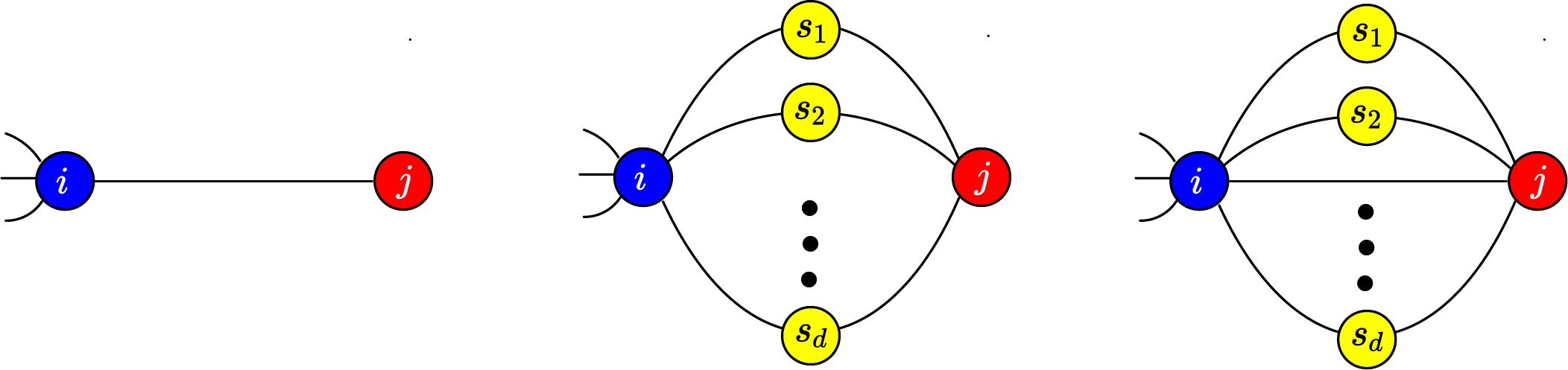}
	\vspace{-2mm}
	\caption{All possible forms of a generalized leaf with head $j$ and tail $i$: (i) When node $i$ is the only neighbor of node $j$; (ii) When node $j$ is not adjacent to node $i$ and only has neighbors of degree two and node $i$ is adjacent to all of them; (iii) The situation of case (ii) while the edge $\{i,j\}$ is also added.}
	\label{leaf}
	\vspace{-4mm}
\end{figure} 

By distinguishing the only case for a full (deterministic) recovery, the next step is to analyse and reduce the partial (stochastic) recovery of private values. Assuming that there are no generalized leaves in the graph, under having Gaussian distribution for all randomness sources in the model, we then show that joint distribution of the elements of $\mathfrak{D}_i^r$ is a non-degenerate Gaussian distribution that has an invertible covariance matrix and compute the mutual information between this joint distribution and $u_j$. This leads to a closed-form expression for $\pi_{i}^{(j)}(t)$ and shows how it can be decreased by the model's tunable parameters. 
 
\section{Privacy Analysis}\label{sec:me pr}
In this section, we first introduce a representation of $\mathfrak{D}_i(t)$ in the form of a matrix-vector decomposition (\cref{md}), and then use it to formally state our results on privacy  (\cref{FS}).

\subsection{Matrix-Vector Decomposition for $\mathfrak{D}_i(t)$}\label{md}
To obtain \cref{informal}, we start by finding a representation of the elements of $\mathfrak{D}_i(t)$ in terms of the randomness sources of our model. To this aim, let $\Gamma\in\R^{n\times n}$ be the matrix whose $ij$-th element is $\Gamma_{ij}$ defined in Algorithm~\ref{alg}. Based on the preparation phase, we have $\mathbf{v}{(0)} = \Gamma^{\T}\mathds{1}_n$ which results in
\begin{equation}\label{a9}
 \mathbf{v}{(t)}= W^{t}\mathbf{v}{(0)} = W^{t}\Gamma^{\T}\mathds{1}_n\quad \Rightarrow \quad	v_i{(t)}=\left[\mathbf{v}(t)\right]_i=  \left[W^{t}\Gamma^{\T}\mathds{1}_n\right]_{i}=\left[W^{t}\right]_{i*}\Gamma^{\T}\mathds{1}_n.
\end{equation}
Consider all the (independent) sources of randomness in our model which are $\setp{u_i}_{i\in[n]}$ and $\setp{\Gamma_{ij}}_{(i,j)\in S}$. Equation~\eqref{a9} implies that $v_i(t)$ can be written as a linear combination of these quantities. The following expression suggests a notation for this linear combination:
\begin{equation}\label{obs1}
	v_{i}{(t)}=\sum_{j \in[n]} \alpha_{j}^i(t)\, u_{j}+\sum_{(k, \ell) \in S} \beta_{k \ell}^{i}(t)\, \Gamma_{k \ell}.
\end{equation}
\eqref{a9} implies that the coefficients in \eqref{obs1} can be obtained in terms of the elements of $\left[W^{t}\right]_{i*}$ as
\begin{equation}\label{coeff}
	\alpha_{j}^i(t) = \left[W^{t}\right]_{i\,m_j}, \quad \beta_{k \ell}^{i}(t) = \left[W^{t}\right]_{i\,\ell}-\left[W^{t}\right]_{i\, m_k}.
\end{equation}
To express \eqref{obs1} in a vector-multiplication form, we define $\mathbf{g}$ to be a $2m$-dimensional (column) vector that is the concatenation of all sources of randomness and $\mathbf{p}_i(t)$ to be the concatenation of the coefficients as follows: (Note that $\mathbf{g},\mathbf{p}_i(t)\in\R^{2m}$.)
\begin{equation}\label{a16}
	\mathbf{g} = \setp{u_{j}}_{j\in[n]}\times\setp{\Gamma_{k\ell}}_{(k,\ell)\in S},\quad \mathbf{p}_i(t) = \setp{\alpha_{j}^i(t)}_{j\in[n]}\times\setp{\beta_{k \ell}^{i}(t)}_{(k,\ell)\in S}.
\end{equation}
Hence, $v_{i}{(t)}$ can be written as 
\begin{equation}\label{vectorize}
	v_{i}{(t)}= \mathbf{p}_i^{\T}(t)\,\mathbf{g}.
\end{equation}
Similar to $v_i(t)$ in \eqref{vectorize}, other elements of $\mathfrak{D}_i(t)$ are also linear combinations of the elements of $\mathbf{g}$ and can be written in vector-multiplication forms in terms of $\mathbf{g}$. To this aim, we can consider matrices $\Delta_i$ and $\Lambda_i$ (where both are of size $\deg_G(i)\times 2m$) such that
\begin{align}
	\label{a2}\setp{u_i}\times\setp{\Gamma_{i \ell}}_{\ell\in N_i \setminus \setp{m_i}} = \Delta_i \, \mathbf{g},\quad\quad  \setp{\Gamma_{\ell i}}_{\ell\in N_i} = \Lambda_i \, \mathbf{g}.
\end{align}
Putting together the representation of the elements of $\mathfrak{D}_i(t)$ in terms of $\mathbf{g}$ in \eqref{vectorize}, and \eqref{a2}, we obtain the matrix $\mathfrak{R}_i(t)$ by vertically concatenating the matrices $\Delta_i$, $\Lambda_i$, and appending vectors $\mathbf{p}_{\ell}(t)$ to the end. More formally, to be compatible with our notation for concatenation of vectors, we write
\begin{equation}\label{a23}
\mathfrak{R}_i(t) =\setp{\left[\Delta_i\right]_{s*}}_{s}\times\setp{\left[\Lambda_i\right]_{s\ast}}_{s}\times\setp{\mathbf{p}_{\ell}^{\top}(\tau)}_{\ell\in N_i, 0\leq \tau\leq t}, \quad \mathfrak{D}_i(t)=\mathfrak{R}_i(t)\mathbf{g}.
\end{equation} 
Recall that $[\cdot]_{s\ast}$ denotes the $s$-th row. Since we set the notation only for the concatenation of vectors, to concatenate matrices $\Delta_i$ and $\Lambda_i$, we first split them into their rows and then concatenate the rows.

\subsection{Formal Statement of Privacy Results}\label{FS}
In this section, we formally state the privacy results using the linear-algebraic representations obtained in \cref{md}. As the first step, we seek to spot the iterations at which node $i$ receives a message that contains new information. To do so, we define $\mathbf{a}_{s}$ as the $s$-th row of $\mathfrak{R}_i(t)$ for some $t\geq s$. We then run a basis-pursuit procedure over $\setp{\mathbf{a}_{s}}_{s=0}^{\infty}$, as described in \cref{pre}. Suppose  $\mathcal{B}(\setp{\mathbf{a}_{s}}_{s=0}^{\infty})$ is the basis obtained by the basis-pursuit scheme. It is straightforward to confirm that the elements $[\Delta_i]_{s*}$ and $[\Lambda_i]_{s*}$, for all $s$, are all linearly independent and lie in this basis. Denote the rest of the elements of the basis by $\mathbf{p}_{j_1}(t_1),\ldots,\mathbf{p}_{j_k}(t_k)$, where $j_1,\ldots,j_k \in N_i$ and $0\leq t_1\leq\ldots\leq t_k$. More formally,
\begin{align}
	\label{a322}\mathcal{B}(\setp{\mathbf{a}_{s}}_{s=0}^{\infty}) = \setp{\left[\Delta_i\right]_{s*}: s} \cup \setp{\left[\Lambda_i\right]_{s*}: s}\cup \setp{\mathbf{p}_{j_1}^{\top}(t_1),\ldots,\mathbf{p}_{j_k}^{\top}(t_k)}.
\end{align}
We want to show that the elements of $\mathcal{B}(\setp{\mathbf{a}_{s}}_{s=0}^{\infty})$ are all the data points at node $i$ that matter and the rest of data points are indeed redundant. To this aim, we first need the following definition.
\begin{definition}\label{def:5}
	Define $((j_1,t_1),\ldots,(j_k,t_k))$  in \eqref{a322} to be the informative sequence of node $i$. Moreover, for a given integer $0\leq t\leq \infty$, define $q(t)$ to be the largest $r\in \{1,\ldots,k\}$ such that $t_r\leq t$.
\end{definition}
Next, let $\mathfrak{D}^r_i$ be the accumulated data that considers the data points corresponding to the informative sequence at node $i$, i.e., the messages coming from neighbors $j_1,\ldots,j_r$ at times $t_1,\ldots,t_r$ for some $r\in\{1,\ldots,k\}$. More formally,  
\begin{align}
	\label{acc data}\mathfrak{D}^r_i = \setp{u_i}\times\setp{\Gamma_{i\ell}}_{\ell \in N_i\setminus \setp{m_i}}\times \setp{\Gamma_{\ell i}}_{\ell\in N_i}\times \left(v_{j_1}(t_1),\ldots, v_{j_r}(t_r)\right).
\end{align}
As a next step, we seek to formally show that the information content of $\mathfrak{D}_i(t)$ equals that of $\mathfrak{D}^{q(t)}_i$, with $q(t)$ defined in Definition~\ref{def:5}.
This claim is proved in the following statement.
\begin{theorem}\label{theorem:1}
	Consider Algorithm~\ref{alg} with $n\geq 2$ under Assumption~\ref{as: 1}. Moreover, consider $q(t)$ from Definition~\ref{def:5}. Suppose $0\leq t\leq \infty$ is given and let $r=q(t)$. Having $\mathfrak{D}_i^r$ defined in \eqref{acc data},  for any $j\neq i$, we have $\mathcal{I}(\mathfrak{D}_i(t)\,;\,u_j) = \mathcal{I}(\mathfrak{D}_i^r\,;\,u_j).$
\end{theorem}
Theorem~\ref{theorem:1} implies that $\mathcal{I}\left(\mathfrak{D}_i(\infty)\,;\,u_j\right) = \mathcal{I}\left(\mathfrak{D}_i^k\,;\,u_j\right)$, suggesting that all but finitely many messages are redundant. The following proposition upperbounds the largest informative time instance $t_k$.
\begin{proposition}\label{lem:3}
	Consider Algorithm~\ref{alg} under Assumption~\ref{as: 1} and let $((j_1,t_1),\ldots,(j_k,t_k))$ be the informative sequence at node $i$. If $\mu_W$ is the minimal polynomial of matrix $W$, then
	\begin{equation}
		t_k \leq \deg(\mu_W)-1\leq n-1.
	\end{equation}
\end{proposition}
Proposition~\ref{lem:3} states that no more information about private values will be exchanged after $n$ rounds of consensus and from then all the messages throughout the network are redundant. Hence, up to $t=n$, it can be determined whether or not an adversary is able to recover a private value. It is also interesting to know how many consensus rounds a (curious) node $i$ must wait to hear about $u_j$ for the first time.
The following proposition is dedicated to this result.
\begin{proposition}\label{lem: 4}
	Consider Algorithm~\ref{alg} under Assumption~\ref{as: 1} and let $((j_1,t_1),\ldots,(j_k,t_k))$ be the informative sequence at node $i$ from Definition~\ref{def:5}. Further, suppose the underlying graph is connected and for the consensus matrix, $W$, we have $W \geq 0$, and $W\sim A_G$. Let $t=\tau_j^{(i)}$ be the minimum $t$ such that $\alpha_j^{i}(t)$ defined in \eqref{coeff} is not zero. Then $\tau_j^{(i)}=d_G(i,m_j)$ and $r(G)-2 \leq \operatorname{ecc}_G(i)-2 \leq t_k$.
\end{proposition}

Analogous to \eqref{a2}, we can define the data collection on non-redundant data based on the concept of the informative sequence at one node and its representation in terms of $\mathbf{g}$ in a matrix-vector multiplication format. To this aim, first consider the informative sequence $((j_1,t_1),\ldots,(j_k,t_k))$ at node $i$ and for $r\in\{1,\ldots,k\}$, define
\begin{equation}\label{ae}
\mathfrak{R}_i^r =\setp{\left[\Delta_i\right]_{s*}}_{s}\times\setp{\left[\Lambda_i\right]_{s\ast}}_{s}\times\left\{\mathbf{p}_{j_{\ell}}^{\top}(t_{\ell})\right\}_{1\leq \ell\leq r},\quad \mathfrak{D}_i^r=\mathfrak{R}_i^r\mathbf{g},
\end{equation} 
and note that $\mathfrak{R}_i^r$ is a full row-rank matrix. \cref{ae} represents all the informative data points at node $i$ in terms of the model's randomness sources, i.e., in terms of the elements of $\mathbf{g}$. In this representation, each column of $\mathfrak{R}_i^r$ corresponds to the coefficients of one randomness source for all data points.  As we show later, the difference between the rank of $\mathfrak{R}_i^r$  before and after removing the column corresponding to a private value determines the privacy-preserving properties of the network. Towards this argument, let $\mathfrak{R}_{i,-j}^r$ be the matrix obtained by removing the column corresponding to $u_j$. We use $-j$ to emphasize that this column is removed. Two cases are possible in this situation: 
\begin{align}
	\label{case1}&\text{Case 1:}\,\,\operatorname{rank}\left(\mathfrak{R}_{i,-j}^r\right)  = \operatorname{rank}\left(\mathfrak{R}_{i}^r\right),\\
	\label{case2}& \text{Case 2:}\,\,\operatorname{rank}\left(\mathfrak{R}_{i,-j}^r\right)  = \operatorname{rank}\left(\mathfrak{R}_{i}^r\right)-1.
\end{align}
We will show that under Case $2$, node $i$ can deterministically recover $u_j$ while under Case $1$, this is not possible. We start by presenting the following lemma stating that under Case $2$, the privacy fails.
\begin{lemma}\label{lemma:case2}
	Consider Algorithm~\ref{alg} with $n\geq 2$ under Assumption~\ref{as: 1}. Moreover, consider $r=q(t)$ from Definition~\ref{def:5} and suppose \eqref{case2} holds for some $j\neq i$. Then node $i$ can fully recover $u_j$ using $\mathfrak{D}^r_i$.
\end{lemma}
Lemma~\ref{lemma:case2} shows that under Case $2$, node $i$ can deterministically recover $u_j$, meaning that $u_j$ is a deterministic function of $\mathfrak{D}_i(t)$. Since we assumed all the random variables in our model are continuous random variables, this leads to $\pi_{i}^{(j)}(t)=\infty$. Next, in \cref{lem:2}, we obtain a closed-form expression for $\pi_{i}^{(j)}(t)$ under Case $1$, which implies that $\pi_{i}^{(j)}(t)$ is finite under Case $1$. This implies that $u_j$ is not a deterministic function of $\mathfrak{D}_i(t)$ and thus node $i$ cannot fully recover $u_j$ using $\mathfrak{D}_i(t)$. Thus, Case $2$ holds if and only if node $i$ can deterministically  recover $u_j$. To analyse Case $1$, we need to define $\mathcal{S}_{-j}$ in the following, which is the covariance matrix of $\mathbf{g}$ when the element $u_j$ is removed:
\begin{equation}\label{covariance}
	\mathcal{S}_{-j} =\operatorname{diag}\left( \setp{\sigma_i^2}_{i\in[n]\setminus \setp{j}}\times\setp{\sigma_{k\ell}^2}_{(k,\ell)\in S}\right).
\end{equation}
Note that the parameters in \eqref{covariance} are previously defined in Definition~\ref{def: parameters}.  The following theorem computes $\pi_{i}^{(j)}(t)$ under Case $1$, for a Gaussian model.
\begin{theorem}\label{lem:2}
	Consider Algorithm~\ref{alg} with $n\geq 2$ under Assumption~\ref{as: 1}. Suppose $0\leq t\leq \infty$ is given and let $r=q(t)$ as defined in Definition~\ref{def:5}. Moreover, suppose \cref{case1} holds and assume that all random variables $u_s,\, s\in[n]$, and $\Gamma_{k\ell},\,(k,\ell)\in S$, have Gaussian distributions. Then, for $j\neq i$ 
	\begin{equation}\label{ad}
		\pi_{i}^{(j)}(t)=\frac{1}{2} \log\left(1+\sigma_j^2\,\ab^{\T}\Sigma^{-1} \ab\right), 
	\end{equation}
	where $
	\ab = \left[\mathfrak{R}_{i}^r\right]_{*j}$, and $\Sigma = \mathfrak{R}_{i,-j}^r\mathcal{S}_{-j}\mathfrak{R}_{i,-j}^{r\T}$. 
\end{theorem}	
Note that for $(k,\ell)\in S$, the quantity $\sigma_{k\ell}^2$, i.e., the variance of the generated noise term $\Gamma_{k\ell}$ lies in the diagonal of $\mathcal{S}_{-j}$. As a result, $\pi_{i}^{(j)}(t)$ is a decreasing function of $\sigma_{k\ell}^2$. Since we are free to choose $\sigma_{k\ell}^2$, we set them large enough to decrease $\pi_{i}^{(j)}(t)$ close to its minimum. However, the cost of increasing $\sigma_{k\ell}^2$ appears in the convergence time. This result is formalised in \cref{conv theorem} and the effect of increasing $\sigma_{k\ell}^2$ on $\pi_{i}^{(j)}(t)$ is also evaluated in the simulations (see \cref{sim}).

As the next step, to avoid Case $2$, we investigate under what conditions Case $1$ and Case $2$ hold. To this aim, the notion of a generalized leaf is introduced.
\begin{definition}[Generalized Leaf]\label{gene}
	Consider graph $G=(V,E)$ with distinct vertices $i,j\in V$ where $\setp{i,j}$ may or may not be in $E$. Then, $G$ has a generalized leaf with head $j$ and tail $i$ if for every $s\in N_j\setminus \setp{i}$, we have $\deg_G(s) = 2$ and $s\in N_i$.
\end{definition}
 A generalized leaf is a generalization of a leaf (i.e., a node with degree $1$). Figure~\ref{leaf} illustrates all possible forms of a generalized leaf with head $j$ and tail $i$. For a generalized leaf of head $j$ and tail $i$, node $i$ can recover $u_j$ in the first iteration of the consensus. To see this, suppose $s\in N_j$. Node $s$ receives $\Gamma_{is}$ and $\Gamma_{js}$ in the preparation phase and initializes its consensus by $\Gamma_{is}+\Gamma_{js}$. Node $i$ knows $\Gamma_{is}$ and has also received $\Gamma_{is}+\Gamma_{js}$ from node $s$ in the first iteration of the consensus, thus it obtains $\Gamma_{js}$. Since node $i$  is adjacent to all neighbors of $j$, node $i$ recovers $\Gamma_{js}$ for all $s\in N_j$. Summing these reveals $u_j$ to node $i$. Next theorem guarantees that this is the only troubling case.
\begin{theorem}\label{theorem:2}
	Consider Algorithm~\ref{alg} over a connected graph $G$. Then \eqref{case1} holds for every distinct pair $i,j\in [n]$ and all $r\in \{1,\ldots,k\}$ if $G$ does not contain any generalized leaf.
\end{theorem}

\section{Convergence Analysis}\label{sec: conv}
Next, we study the convergence properties of Algorithm~\ref{alg}. We first formalize the notion of convergence time based on the waiting time required to achieve an $\epsilon$-accurate estimation of the convergence limit. We denote this quantity  by $t_{\epsilon}$ and define it as follows.
\begin{definition}\label{conv time}
	For $\epsilon>0$ and $\mathbf{v}(t)$ in \eqref{a9}, let $t_{\epsilon}$ be the minimum time $t$ such that $\norm{\mathbf{v}(t)-\mathbf{u}^*}_2 \!\leq \!\epsilon$.
\end{definition}
Note that the convergence rate \eqref{conv rate} is conceptually different from the convergence time defined in Definition~\ref{conv time}. While the rate represents the ''slope" of a convergence, the convergence time refers to the time when the iterations arrive the vicinity of the solution. This for example, implies that a larger distance between the initial values and the final solution leads to a larger convergence time under the same convergence rate. Having this definition, the convergence results of this paper is summarized in the following theorem.
\begin{theorem}\label{conv theorem}
	Consider Algorithm~\ref{alg} (PPAC) under Assumption~\ref{as: 1}. Further, suppose $W\mathds{1}_n = W^{\T}\mathds{1}_n = \mathds{1}_n$ and $\rho( W - {1}/{n}\mathds{1}_n\mathds{1}_n^{\T}) <1$. Then, we have:
	\begin{enumerate}[(i)]
		\item In the limit, PPAC converges to the exact solution, i.e., $\lim_{t \to \infty } \mathbf{v}(t) = \mathbf{u}^*$.
		\item \label{gt}The convergence rate in part (i) is the same as the convergence rate of an ordinary (non-private) consensus, meaning that \eqref{conv rate} holds for PPAC too.
		\item \label{gt2} The expected convergence time for achieving an $\epsilon$-accurate solution is 
		\begin{equation}\label{ad2}
			\E\left[t_{\epsilon}\right]=\mathcal{O}\left(\log\left(\frac{1+\mu_{\max}^2+\sigma_{\max}^2}{\epsilon}\right)\right),
		\end{equation}
		where $\mu_{\max}$ and $\sigma_{\max}^2$ are from Definition~\ref{def: parameters}.
	\end{enumerate}
\end{theorem}
Note that the conditions $W\mathds{1}_n = W^{\T}\mathds{1}_n = \mathds{1}_n$ and $\rho( W - {1}/{n}\mathds{1}_n\mathds{1}_n^{\T}) <1$, as mentioned in \cref{pre} are essential for the fundamental convergence properties of the classical averaging consensus. Further, we refer to \cref{sim} for simulation results on convergence. 
\begin{remark}
	The trade-off between privacy and the convergence time in our model can be explained as follows. Earlier in this section, we mentioned that increasing $\sigma_{k \ell}^2$ for $(k,\ell)\in S$ (or similarly $\sigma_{\max}^2$) decreases the privacy leakage $\pi_{i}^{(j)}(t)$ due to \cref{ad}, while it increases the convergence time due to \cref{ad2}. Hence, the quantities $\sigma_{k \ell}^2$ for $(k,\ell)\in S$ can be seen as the tuning parameters that adjust the position in the trade-off between privacy and the convergence time. 
\end{remark}

\section{Conclusion}

In this work, we introduced a novel private consensus
averaging algorithm and analyzed its privacy from an
information-theoretic perspective. We specifically
used the mutual information function between all the
information available at the adversary node and the
initial value of the victim to measure how private the
value of the victim is with respect to the attacker.
Our results show that the convergence rate of our proposed method is similar to the convergence rate of a
classic non-private consensus method, for any level of
privacy. More importantly, any level of accuracy can
be achieved via our proposed method. Our results also
show that there exists a trade-off between the level of
information-theoretic privacy and convergence time.

\section*{Acknowledgment}
The work of M. Fereydounian and H. Hassani is funded by DCIST, NSF CPS-1837253, and NSF CIF-1943064 and  NSF CAREER
award CIF-1943064, and Air Force Office of Scientific Research Young Investigator Program
(AFOSR-YIP) under award FA9550-20-1-0111. 
The research of A. Mokhtari is supported in part by NSF Grants 2019844 and 2112471, ARO Grant
W911NF2110226, the Machine Learning Lab (MLL) at UT Austin, and the Wireless Networking and Communications Group (WNCG) Industrial Affiliates Program.
 The work of R. Pedarsani is funded by NSF Grant 2003035.

\bibliography{ref}

\newpage
\appendix

% !Tex root = document.tex

\section{Simulations}\label{sim}
In this section, we provide simulations as a visualization to theoretical results of this paper. Consider a private consensus problem with the underlying graph $G$ consisting of $n=6$ nodes as shown in Figure~\ref{graph}. The vector of private values is denoted by $\mathbf{u}=[u_1,\ldots,u_6]^{\T}$, where each component is independently generated from the normal distribution with mean $0$ and standard deviation $\sigma_0 =10$, i.e., $u_i \sim \mathcal{N}(0,\sigma_0^2)=\mathcal{N}(0,100)$ for all $i\in [6]$. The specific realization of this example was $u =[2.30,\,4.40,\,-6.17,\,2.75,\,6.01,\,0.92]^{\T}$ with average $u^* = 1.7$. We then ran Algorithm~\ref{alg} (PPAC) over this model. For each node $i$, we chose $m_i$ uniformly at random in $N_i$ and the realization for this example is $(m_1,\ldots,m_6)=( 2,4,4,3,4,1)$. All the pure noises are independently generated from normal distribution with mean $0$ and standard deviation $\sigma_N $, i.e., $\Gamma_{ij} \sim\mathcal{N}(0,\sigma_N^2)$ for all $(i,j)\in S$. In all cases, we use the following consensus matrix: $W=I_n-1/d_{\max}(D_G-A_G)$, where $I_n$ is the identity matrix, $D_G=\operatorname{diag}(\deg_G(1),\ldots,\deg_G(n))$ and $d_{\max}=\max_{i\in [n]} \deg_G(i)$.

First, one might be interested in observing a realization of the pure noises and the corresponding dynamics (convergence) of nodes' messages for such a realization. For this purpose, we set $\sigma_N=15$. Note that PPAC generates new initial values for the consensus different from $\mathbf{u}$ but with the same average. Under the aforementioned realization, the new initial values were $\mathbf{v}(0) = [v_1(0),\ldots,v_6(0)]^{\T}=  [-12.57,\,16.49,\,-20.11,\,20.78,\,-23.08,\,28.70]^{\T}$. Figure~\ref{node conv} shows the convergence of $v_i(t)$ for all nodes in terms of the iteration number (time) $t$. As it can be seen, all nodes converge to the same value, which is the true average $u^*=1.7$. 
\begin{figure}[b!]
	\begin{minipage}{0.48\linewidth}
	\centering
		\includegraphics[width=0.7\textwidth]{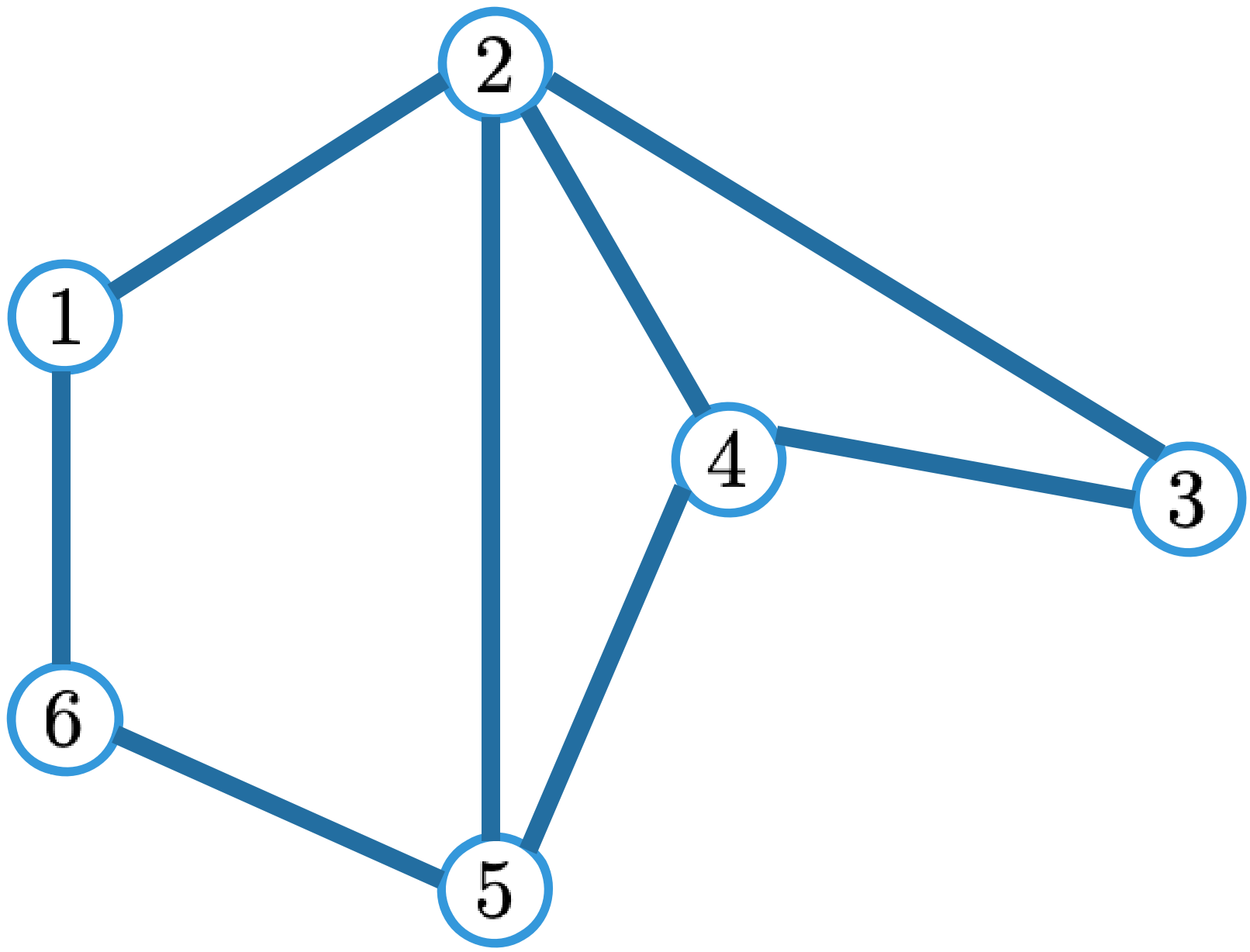}
		\caption{The graph $G$}
		\label{graph}
	\end{minipage}
	\hspace{3mm}
	\begin{minipage}{0.48\linewidth}
	\centering
		\includegraphics[width=\textwidth]{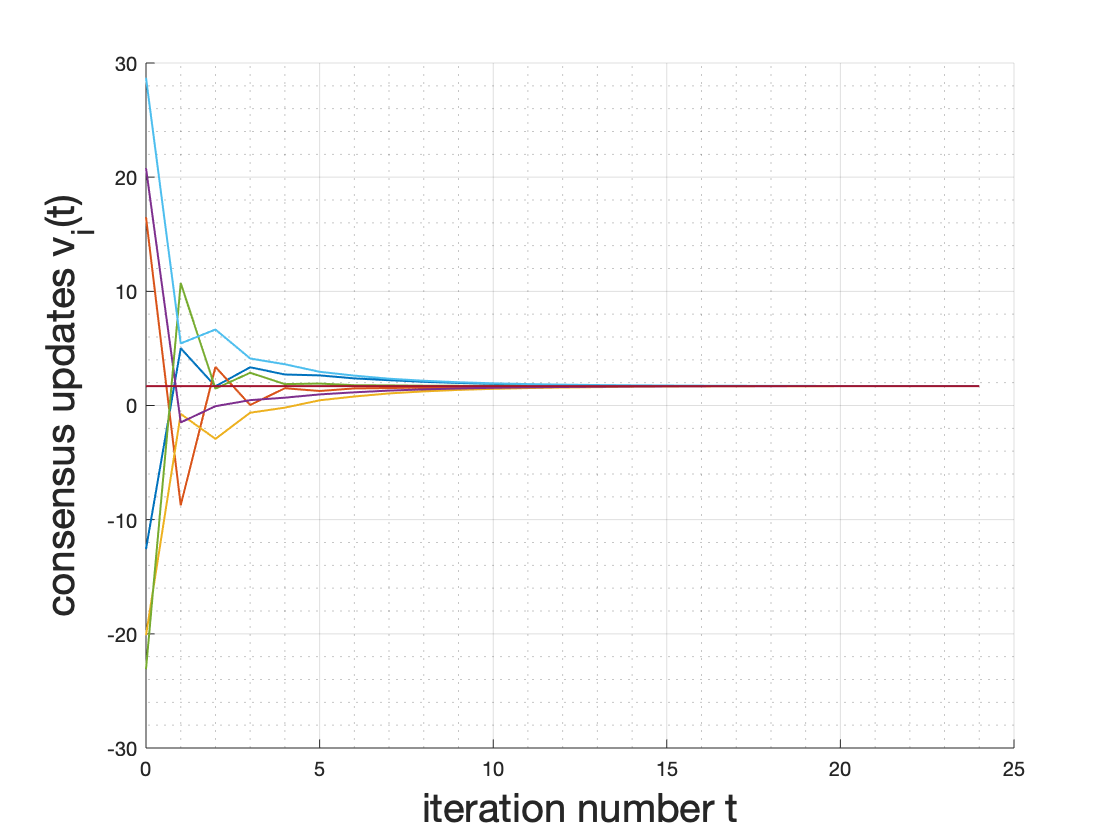}
		\caption{One realization of the convergence of $v_i(t)$ in terms of $t$ for all nodes using $\sigma_N=15$}
		\label{node conv}
	\end{minipage}
\end{figure}

The second  plot of interest answers the following question: While the initial values are kept fixed, how does the total convergence error over time behave when we increase the variance of pure noises $\sigma_{k\ell}^2$ generated by Algorithm~\ref{alg}? To this aim, we plot the total convergence error, i.e., $\norm{\mathbf{v}(t)-\mathbf{u}^*}_2$ in terms of time $t$ for different values of $\sigma_N$. We consider a wide range of values, that is $\sigma_N = 15,\, 150,\, 1500,\,15000,\,150000,\,1500000$. Recall that the generated noises $\Gamma_{ij}, \, i\in [n], j\neq m_i$ are realizations from the normal distribution with mean zero and variance $\sigma_N^2$ and thus the quantity $\norm{\mathbf{v}(t)-\mathbf{u}^*}_2$ is affected by particular realization of the values $\Gamma_{ij}$. To release the plots from this randomness, we generated $100$ realizations of $\Gamma_{ij}$ and averaged $\norm{\mathbf{v}(t)-\mathbf{u}^*}_2$ for each $t$ over these realizations. Figure~\ref{conv norm} shows the resulted plots. As it can be seen, Figure~\ref{conv norm} agrees with Theorem~\ref{conv theorem} and shows that increasing the variance of the generated noises $\Gamma_{ij}, \, (i,j)\in S$ affects the convergence time $t_{\epsilon}$ of achieving a given error $\epsilon$ only by a logarithmic factor. Note that the $y$-axis is in logarithmic scale.  
\begin{figure}[t!]
	\begin{minipage}{0.48\linewidth}
	\centering
		\includegraphics[width=\textwidth]{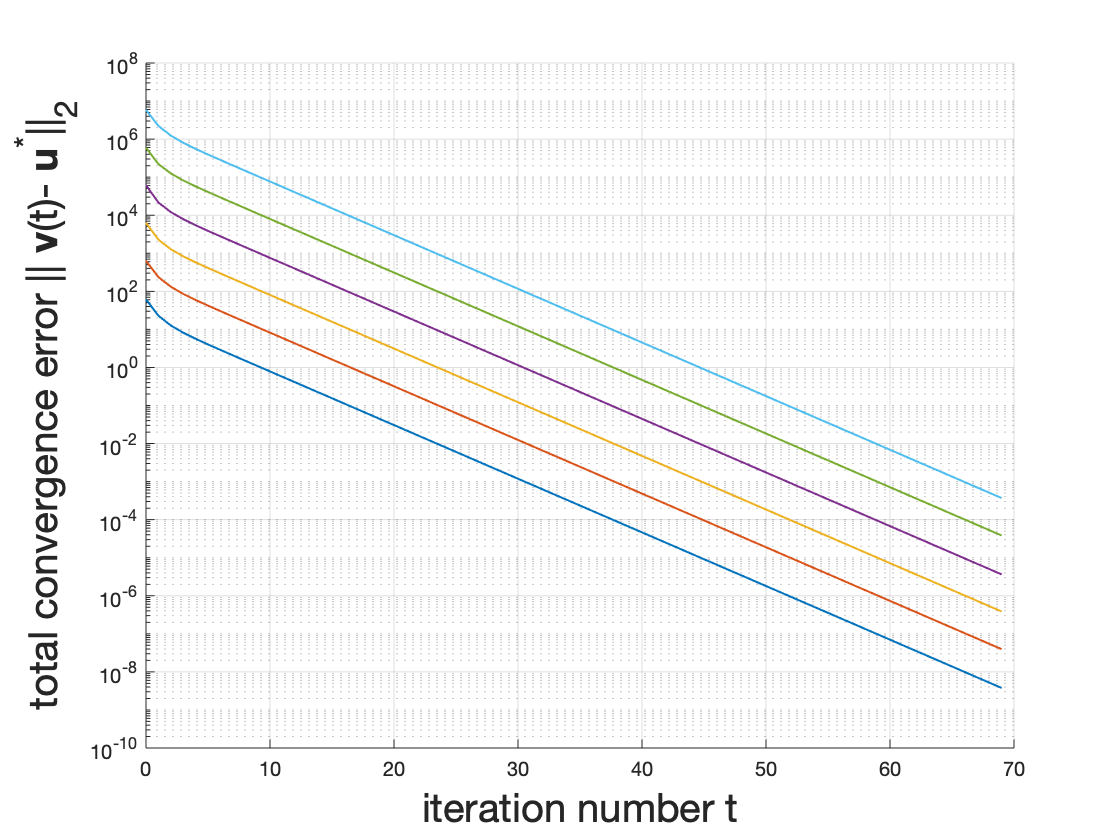}
		\caption{The total convergence error $\norm{\mathbf{v}(t)-\mathbf{u}^*}_2$ in terms of iteration number $t$ for a wide range of generated noises' standard deviation $\sigma_N$, i.e., for values $\sigma_N = 15,\, 150,\, 1500,\,15000,\,150000,\,1500000$}
		\label{conv norm}
	\end{minipage}
	\hspace{3mm}
	\begin{minipage}{0.48\linewidth}
	\centering
		\includegraphics[width=\textwidth]{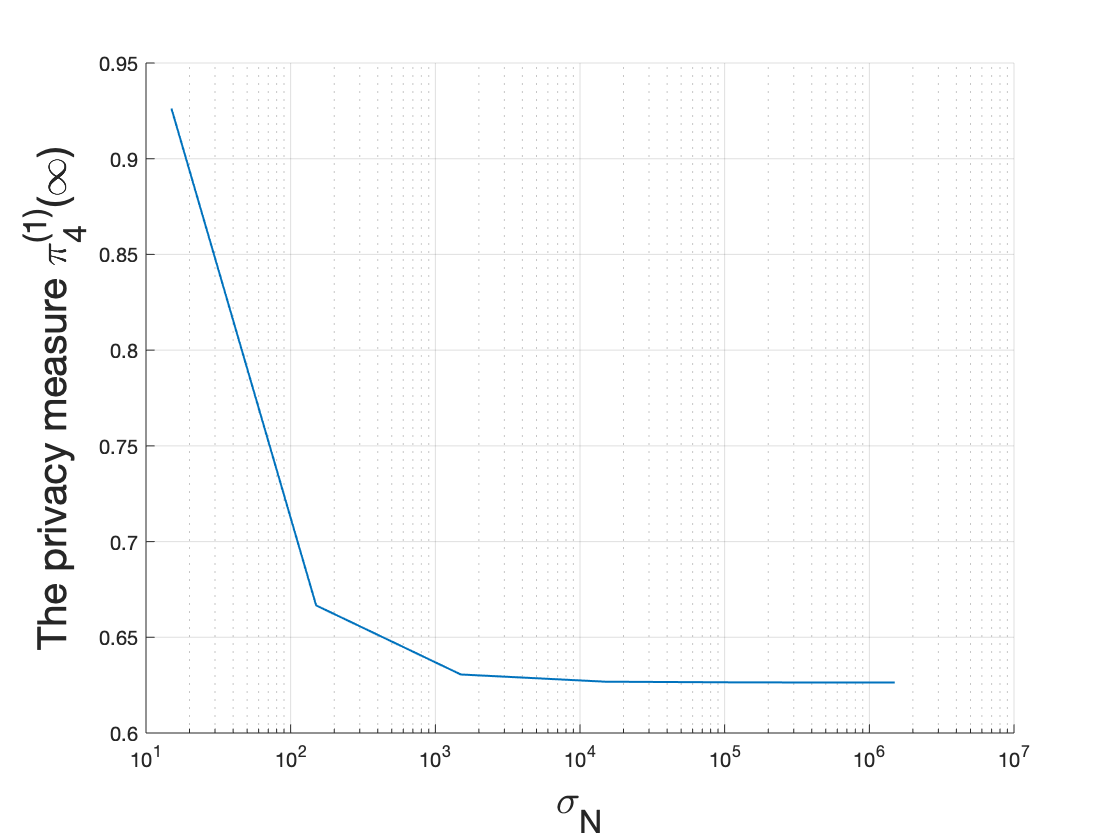}
		\caption{The leakage of node $1$ at node $4$, i.e., $\pi_{4}^{(1)}(\infty)$ in terms of generated noises' standard deviation $\sigma_N$ for values $\sigma_N = 15,\, 150,\, 1500,\,15000,\,150000,\,1500000$}
		\label{measure}
	\end{minipage}
\end{figure}
\begin{figure}[h!]
	\begin{minipage}{0.45\linewidth}
	\centering
		\includegraphics[width=0.7\textwidth]{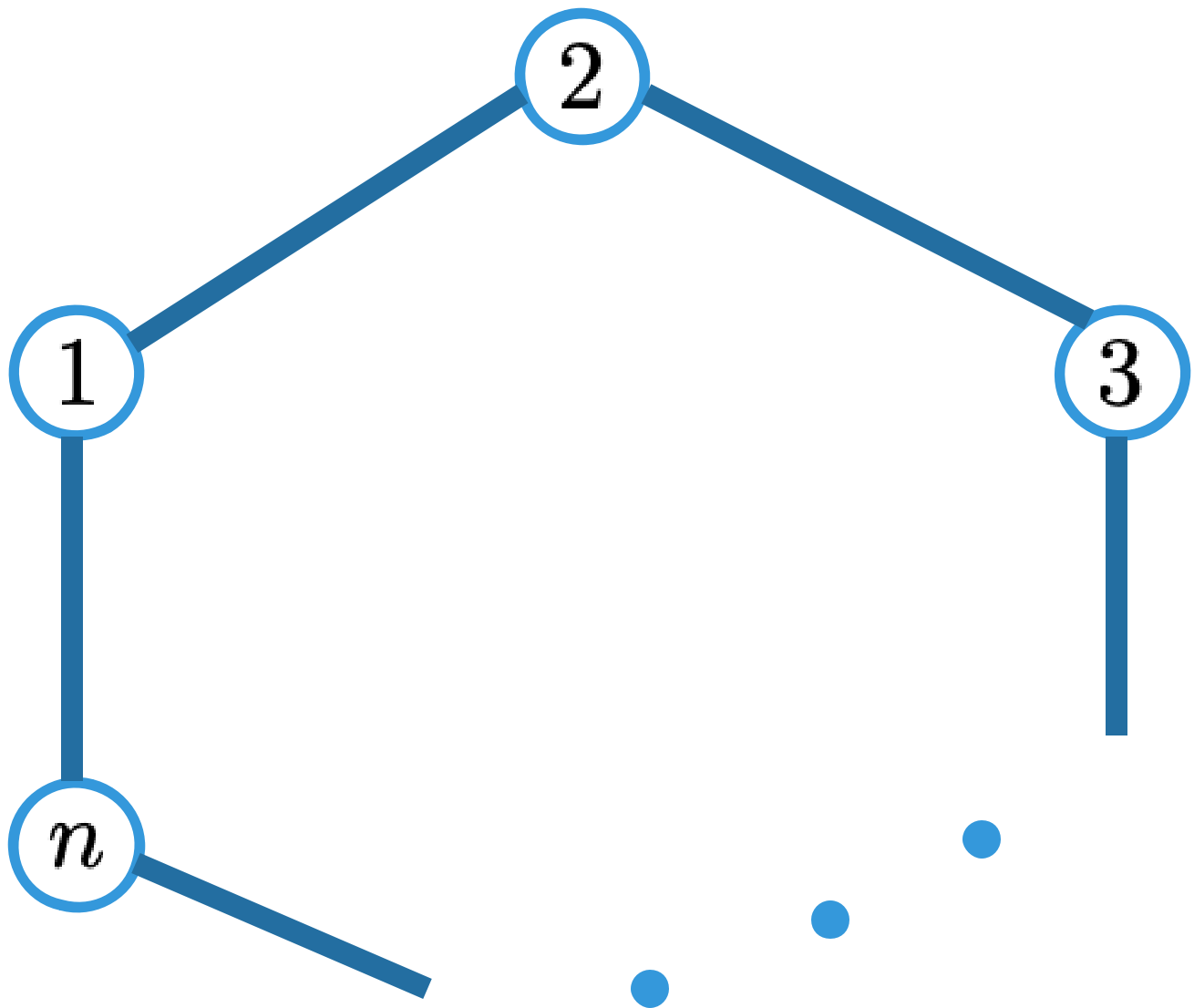}
		\caption{The graph $G=C_n$}
		\label{graph2}
	\end{minipage}
	\begin{minipage}{0.52\linewidth}
	\centering
		\includegraphics[width=\textwidth]{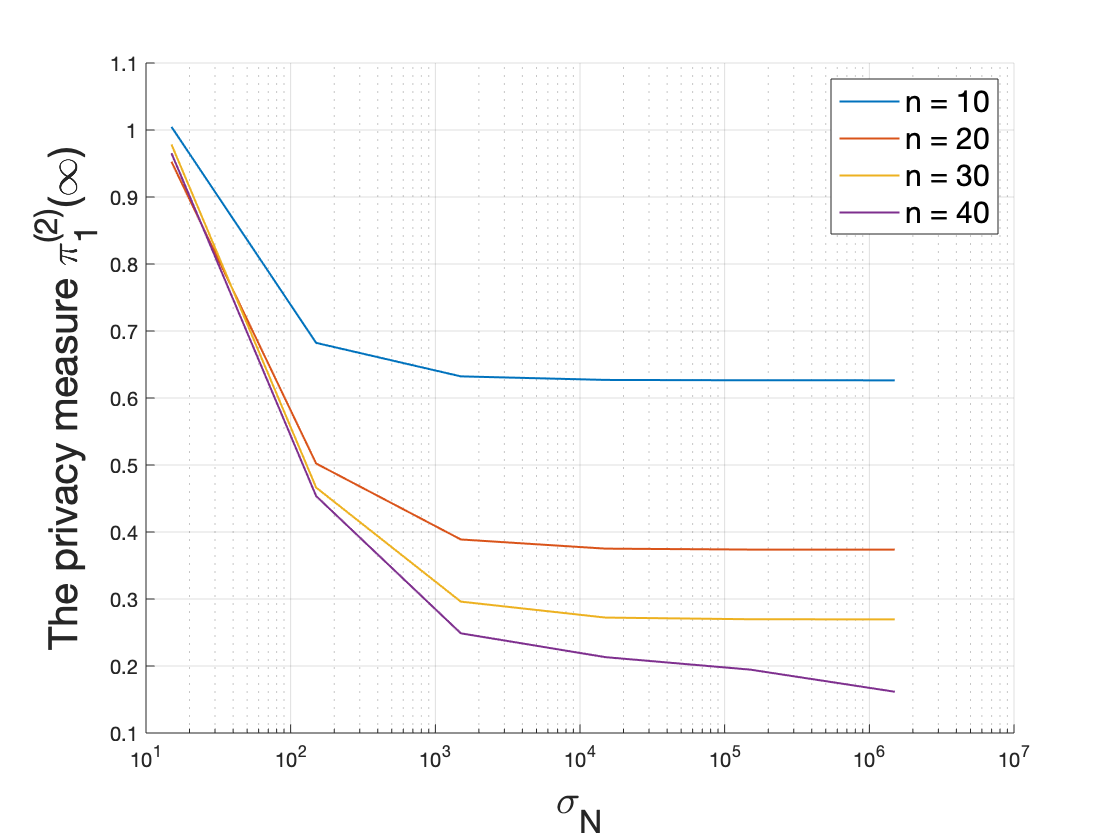}
    	\caption{The leakage of node $2$ at node $1$ in $G=C_n$, i.e., $\pi_{1}^{(2)}(\infty)$ in terms of generated noises' standard deviation $\sigma_N$ for values $\sigma_N = 15,\, 150,\, 1500,\,15000,\,150000,\,1500000$ for values $n = 10,\,20,\, 30,\, 40$}
    	\label{effect}
	\end{minipage}
\end{figure}

As the next step, we run experiments to visualize our privacy results. Note that graph $G$ in Figure~\ref{graph} does not contain any generalized leaf that are illustrated in Figure~\ref{leaf}. More generally, it is interesting to know that any graph over $n\geq 5$ nodes that has a cycle of length $n$ cannot contain a generalized leaf and is therefore a private structure. Hence, due to \cref{informal}, the privacy of every node is preserved at every other node in this example. Having this, one might be interested in computing the leakage measure $\pi_i^{(j)}(t)$ defined by Definition~\ref{def: pr}. This quantity can be computed based on Theorem~\ref{lem:2}. We are interested in considering all the data that one node can collect over all times through Algorithm~\ref{alg}. Therefore, we consider $\pi_i^{(j)}(\infty)$, i.e., the maximum leakage of node $j$ at node $i$. Here, we consider node $i=4$ is curious about $j=1$. Recall that in this experiment, we let all the generated noises have a common variance $\sigma_N^2$. Due to  \eqref{ad}, we know that increasing $\sigma_N^2$ decreases $\pi_i^{(j)}(\infty)=\mathcal{I}\left(\mathfrak{D}_i(\infty)\,;\,u_j\right)$ through a logarithmic function. This is plotted in Figure~\ref{measure} for $\pi_1^{(4)}(\infty)$. Note that the $x$-axis is in logarithmic scale. Few important items are worth mentioning about this figure which are discussed in the following: \\
First, as it can be observed, there is a ceiling on the achievable amount of privacy. This is due to the fact that some amount of information will unavoidably be transmitted by the construction of consensus model. For instance, the average of private values to which all nodes will converge, contains some amount of information about each private value. 
%One can confirm that $\mathcal{I}(u^*;u_j)= 0.5\log(1+1/(n-1))$. The maximum privacy level in our example considering the attacker $i=4$ and victim $j=1$ is about $0.83$. The gap between $0.83$ and $1$ cannot be filled unless we prevent nodes from achieving the exact average. \\
Second, since the scale is logarithmic, choosing smaller values of $\sigma_N$ can provide almost the same privacy level that extremely large values can but smaller values of $\sigma_N$ provide better convergence time as discussed in Figure~\ref{conv norm}. Hence, one might choose  $\sigma_N$ to be smallest value for which the plot in Figure~\ref{measure} is almost flat. $\sigma_N=1500$ seems a good option.
%Here, choosing $\sigma_N=1500$ seems a good option. On one hand, it provides almost $83\%$ privacy and on the other hand it provides fast convergence. 

Following the above discussion, we know that the consensus problem inherently reveals some amount of information. For instance, the limiting value $u^*$ which is achievable by all nodes partially reveal information about each $u_i$. We know that part of the information is leaked due to achieving the exact average $u^*$. This can be computed as $\mathcal{I}(u^*;u_j)= 0.5\log(1+1/(n-1))$ in the case where $u_i$ s are Gaussian. This expression also suggests that when $n$ is large, less information will be naturally be revealed by $u^*$. One might ask if PPAC also provides a smaller leakage for a larger $n$ and thus higher levels of privacy are achievable by PPAC as $n$ gets larger. It turns out that the answer is a yes. To numerically evaluate this notion, we consider the underlying graph to be a cycle of $n$ nodes $G=C_n$ (see Figure~\ref{graph2}) and we suppose one node is curious about its neighbor. For the sake of clarity choose node $1$ as the curious node (attacker) and node $2$ as the victim. Any two neighbors can be chosen due to the symmetry. Figure~\ref{effect} shows a similar plot as in  Figure~\ref{measure} for $G=C_n$, evaluating $\pi_{1}^{(2)}(\infty)$ in terms of $\sigma_N$ when the total number of nodes $n$ is increasing, i.e., $n=10,\,20,\, 30,\, 40$. As it can be seen in Figure~\ref{effect}, larger networks provide lower levels of privacy leakage.

\section{Proofs}

\subsection{Further Notations}
\begin{itemize}
	\item For a matrix $A$, $|A|=|\det(A)|$ and for a set $A$, $|A|$ is the number of elements of $A$.
	\item For a matrix $A$, $A\succ 0$ means that $A$ is a positive definitive matrix. 
	\item For (column) vectors $\ab_{1},\ldots,\ab_s$,  $\left(\ab_{1},\ldots,\ab_s\right)$ denotes the matrix with these columns. 
	\item For $a\in\R$, $\ceil*{a}$ denotes the smallest integer greater than or equal to $a$.
	\item $X\perp Y$: $X$ and $Y$ are statistically independent. 
	\item For a multidimensional random variable $\X$: $\Sigma_{\X}=\operatorname{Cov}(\X)$ is the covariance matrix of $\X$.
	\item For a one-dimensional random variable $X$: $\sigma_X^2 = \operatorname{Var}(X)$ denotes the variance of $X$.
	\item The indicator function with condition $C$ is denoted by $\mathbf{1}_{C}$. Therefore, $\mathbf{1}_{C} =1$ if $C$ holds and $\mathbf{1}_{C} =0$ otherwise.
\end{itemize}

\subsection{Theorem~\ref{theorem:1}}
We first formally prove the fundamental property of a basis-pursuit procedure. Then we express a lemma stating that redundant data can be ignored when computing mutual information. Equivalently, if there is a spanning set that generates all other arguments of a mutual information function, considering this spanning set is enough and other elements can be ignored. The proof is based on the fact that the basis resulted from a basis-pursuit procedure is an example of such spanning sets. We start by formally proving fundamental properties of a basis-pursuit procedure.

\begin{lemma}\label{lemma:basis_persuit}
	Consider $\setp{\ab_i}_{i=0}^{\infty}$ where $\ab_i\in\Rn$ with $\ab_0\neq \mathbf{0}$ and let $B = \mathcal{B}\left(\setp{\ab_i}_{i=0}^{\infty}\right)$ be obtained from a basis pursuit procedure. Then, $|B|<\infty$ and the elements of $B$ are linearly independent. Let $B=\setp{\ab_{t_1},\ldots,\ab_{t_k}}$ where $0=t_1<t_2<\ldots<t_k$. Then for every $i\geq 0$, letting $r$ be the largest number such that $t_r \leq i$, we have
	\begin{equation}\label{r}
	\ab_{i} \in \operatorname{span}\left(\setp{\ab_{t_1},\ldots,\ab_{t_r}}\right).
	\end{equation} 
\end{lemma}
\begin{proof}
	By construction of $\mathcal{B}(\cdot)$, any finite subset of $B$ is linearly independent. Thus, $|B| \leq \dim \left(\Rn\right) = n$. Moreover, if $t_r = i$, \eqref{r} trivially holds. Therefore, suppose $t_r < i$ and $\ab_{i} \notin \operatorname{span}\left(\setp{\ab_{t_1},\ldots,\ab_{t_r}}\right)$, then by construction of $\mathcal{B}(\cdot)$, we have $\ab_{i} \in B$ and thus $i=t_{\ell}$ for some $\ell>r$ which contradicts the definition of $r$.
\end{proof}

As the next step, we prove that a spanning subset of a dataset contains all the information of that dataset.
\begin{lemma}\label{lemma:finite}
	Suppose $\Y$ is an $n$-dimensional and $X$ is an arbitrary random variable. Further, suppose $\ab_1,\ldots,\ab_r\in \Rn$ are constant and for $s\leq r$, let $\ab_1,\ldots,\ab_r\in\operatorname{span}\left(\ab_{1},\ldots,\ab_{s}\right)$. Then
	\begin{equation}
	\mathcal{I}\left(\ab_1^{\T}\Y,\ldots,\ab_r^{\T}\Y\,;\,\, X \right) = \mathcal{I}\left(\ab_1^{\T}\Y,\ldots,\ab_s^{\T}\Y\,;\,\, X \right).
	\end{equation}
\end{lemma}
\begin{proof}
	If $s=r$ the statement is obvious. Suppose $s<r$ and note that since $\ab_1,\ldots,\ab_r\in\operatorname{span}\left(\ab_{1},\ldots,\ab_{s}\right)$, for every $i\in\setp{s+1,\ldots,r}$ there exits $\lambda_{i1},\ldots,\lambda_{is}\in\R$ such that $\ab_i = \lambda_{i1}\ab_1+\ldots+\lambda_{is}\ab_{s}$. Let $\Lambda$ denote the matrix with elements $\lambda_{ij}$ and write
	\begin{equation}
	\left(\ab_{s+1},\ldots,\ab_r\right) = 	\left(\ab_{1},\ldots,\ab_s\right)\Lambda^{\T}.
	\end{equation}
	Hence,
	\begin{align}
	\left(\ab_{s+1}^{\T}\Y,\ldots,\ab_r^{\T}\Y\right) &= \Y^{\T}\left(\ab_{s+1},\ldots,\ab_r\right)	\\ &=\Y^{\T}\left(\ab_{1},\ldots,\ab_s\right)\Lambda^{\T} \\&= \left(\Y^{\T}\ab_{1},\ldots,\Y^{\T}\ab_s\right)\Lambda^{\T} \\\label{f1}&= f\left(\ab_{1}^{\T}\Y,\ldots,\ab_s^{\T}\Y\right).
	\end{align}
	Let $\Z=	\left(\ab_{1},\ldots,\ab_s\right)$. From \eqref{f1}, $\left(\ab_{s+1}^{\T}\Y,\ldots,\ab_r^{\T}\Y\right) = f(\Z)$, where $f(\cdot)$ is a deterministic function. This results in
	\begin{align}
	\mathcal{I}\left(\ab_1^{\T}\Y,\ldots,\ab_r^{\T}\Y\,;\,\, X \right) &= \mathcal{I}\left(\ab_1^{\T}\Y,\ldots,\ab_s^{\T}\Y,\ab_{s+1}^{\T}\Y,\ldots,\ab_r^{\T}\Y\,;\,\, X \right) \\&= \mathcal{I}\left(\Z,f(\Z)\,;\,\, X \right) \\
	&=\mathcal{I}\left(\Z\,;\,\, X \right)\\
	&=\mathcal{I}\left(\ab_1^{\T}\Y,\ldots,\ab_s^{\T}\Y\,;\,\, X \right).
	\end{align}
\end{proof}

\begin{proof}[{\bf Proof of Theorem~\ref{theorem:1}}]
	Consider $\mathfrak{R}_i(t)$ from \eqref{a2}. Then $\ab_s $ is defined to be the $s$-th row of $\mathfrak{R}_i(\infty)$ or more rigorously, the $s$-th row of $\mathfrak{R}_i(t)$ for some $t\geq s$. Moreover, consider the related basis from \eqref{a322}. When finding  $\mathcal{B}(\{\mathbf{a}_{s}\}_{s=0}^{\infty})$, the basis-pursuit procedure adds to the basis all the rows of matrices $\Lambda_i$ and $\Delta_i$. This holds because one can confirm that in the concatenation of these two matrices, i.e.,
	\begin{equation}
	\left[
		\begin{array}{c}
			\Delta_i \\
			\Lambda_i
		\end{array}\right],
	\end{equation}
	each column has at most $1$ non-zero value (the non-zero values are either $1$ or $-1$) and each row has at least one non-zero value. The basis pursuit then continues by evaluating $\mathbf{p}_{\ell}(t)$ for $\ell\in N_i$ and $t\geq 0$ and chooses $\mathbf{p}_{j_1}(t_1),\ldots,\mathbf{p}_{j_k}(t_k)$, where $((j_1,t_1),\ldots,(j_k,t_k))$ is the informative sequence of node $i$ (see Definition~\ref{def:5}). Define
	\begin{equation}
		A_r = \setp{\left[\Delta_i\right]_{s*}: 1\leq s\leq \deg_G(i)} \cup \setp{\left[\Lambda_i\right]_{s*}: 1\leq s \leq \deg_G(i)}\cup \setp{\mathbf{p}_{j_1}^{\top}(t_1),\ldots,\mathbf{p}_{j_r}^{\top}(t_r)}.
	\end{equation}
	Using the notation $\ab_s$, $\mathfrak{D}_i(t)$ and $\mathfrak{D}_i^r$ (for every $r\leq k$), can be written as
	\begin{equation}\label{a24}
		\mathfrak{D}_i(t) = \setp{\mathbf{a}_s^{\T}\mathbf{g}}_{s\in \{0,\ldots,t+2\deg_G(i)\}}, \quad \mathfrak{D}_i^r = \setp{\mathbf{a}_s^{\T}\mathbf{g}}_{\mathbf{a}_s\in A_r}.
	\end{equation}
	Due to Lemma~\ref{lemma:finite}, for $0\leq s \leq t+2\deg_G(i)$ and $r=q(t)$, 
	\begin{equation}\label{a25}
		\mathbf{a}_s \in \operatorname{span}\left(A_r\right).
	\end{equation}
	Having \eqref{a24} and \eqref{a25}, the result follows by applying Lemma~\ref{lemma:finite} in which $\mathbf{Y}=\mathbf{g}$ and $X=u_j$. 
\end{proof}

\subsection{Proposition~\ref{lem:3}}
\begin{proof} Let $\kappa=\deg(\mu_W)$ and denote the coefficients of the minimal polynomial $\mu_W(\cdot)$ by  
	\begin{align}
		\mu_W(X) = X^\kappa + a_{\kappa-1}X^{\kappa-1}+\ldots+a_0I. 
	\end{align}
	Having $\mu_W(W) =0$, $W^{\kappa}$ can be written as a linear combination of $W^0,\ldots,W^{\kappa-1}$. Using induction, this also holds for any $W^t$ when $t\geq \kappa$. Therefore, for every integer $t\geq 0$, there exists $\tilde{a}_0,\ldots,\tilde{a}_{\kappa-1}$ such that 
	\begin{align}\label{a13}
	W^t = \sum_{s=0}^{\kappa-1}\tilde{a}_{s}W^{s}. 
	\end{align}
	Having \eqref{coeff} and \eqref{a13}, we conclude that for any $t\geq 0$ and every $q,j,k,\ell\in [n]$ with $j\neq q$ and $(k,\ell)\in S$: 
	\begin{equation}\label{a15}
	\alpha_{j}^q(t) = \sum_{s=0}^{\kappa-1}\tilde{a}_{s}\alpha_{j}^q(s), \quad \beta_{k \ell}^{q}(t) = \sum_{s=0}^{\kappa-1}\tilde{a}_{s}\beta_{k \ell}^{q}(s).
	\end{equation}
	Having the definition of $\mathbf{p}_q(t)$ in \eqref{a16}, for every $q\in[n]$ and $t\geq 0$, \eqref{a15} results in
	\begin{align}\label{a26}
		\mathbf{p}_q(t) = \sum_{s=0}^{\kappa-1}\tilde{a}_{s}\mathbf{p}_q(s).
	\end{align}
	Note that \eqref{a26} holds for every $q\in [n]$, particularly when $q\in N_i$. Considering $q\in N_i$, \eqref{a26} leads to the fact that for very $t \geq 0$ and $\ell\in N_i$, 
	\begin{equation}
		\mathbf{p}_{\ell}(t) \in \operatorname{span}\left(\setp{\mathbf{p}_{q}(\tau):\, q\in N_i,\, 0\leq \tau \leq \kappa-1}\right). 
	\end{equation}
	Hence, $t_k \leq \kappa-1 = \deg(\mu_{W})-1$. Moreover, due to Cayley-Hamilton theorem, for any $n \times n$ square matrix $W$, we have $\deg(\mu_{W}) \leq n$. 
\end{proof}

\subsection{Proposition~\ref{lem: 4}}
\begin{proof}
	The statement has two parts and the proof will be enumerated accordingly.
	\begin{enumerate}[(i)]
		\item For the adjacency matrix $A_G$ of a simple graph $G$, $\left[A_G^t\right]_{ij}$ has a combinatorial meaning. It is equal to the number of walks of length $k$ between the nodes $i$ and $j$. For a consensus matrix $W\sim A_G$, $W$ can be obtained from $A_G$ by replacing the $1$s and (possibly) the diagonal elements of $A$ with some non-negative numbers. Similar to $\left[A_G^t\right]_{ij}$, $\left[W^t\right]_{ij}$ can also be represented combinatorially. To this aim, define $\Omega_t(i,j)$ to be the set of walks of length $t$ between $(i,j)$, that is
	\begin{align}
		\Omega_t(i,j) = \setp{(s_0,\ldots,s_{t})\in [n]^{t+1} \mid s_0=i, s_t=j, \forall q\in [t]:\, \setp{s_{q},s_{q-1}} \in E }.
	\end{align}
 	Based on the definition of matrix multiplication, we can write
	\begin{align}\label{a17}
		\left[W^t\right]_{ij} = \sum_{(s_0,\ldots,s_{t})\in \Omega_t(i,j)} \prod_{q=0}^{t-1} \left[W\right]_{s_qs_{q+1}}.
	\end{align} 
	Having \eqref{a17}, if $t<d_G(i,m_j)$, then there is no walk of length $t$ between $i$ and $m_j$, i.e., $\Omega_t(i,m_j)=\emptyset$ and $\alpha_j^{i}(t)=\left[W^t\right]_{im_j}=0$. Now for $t=d_G(i,m_j)$, there exists at least one positive term (corresponding to each path of length $d_G(i,m_j)$) in \eqref{a17} which results in 	$\alpha_j^{i}(t)=\left[W^t\right]_{im_j}>0$. Thus, $\tau_j^{(i)}$ exists and equals $d_G(i,m_j)$. This completes the first part.

	\item Suppose $j_0 = \arg\max_{s\in [n]} d_G(i,s)$ which leads to $d_G(i,j_0)= \max_{s\in [n]} d_G(i,s) =\operatorname{ecc}(i)$. If $d_G(i,j_0) \leq 2$, the result is trivial. Suppose $d_G(i,j_0) \geq 3$. It suffices to prove that $t_k \geq d_G(i,j_0)-2$. Define $i_0 = \arg\min_{s\in N_i} d_G(s,m_{j_0})$. We claim that $t_k \geq d_G(i_0,m_{j_0})$. The claim then leads to $t_k \geq d_G(i_0,m_{j_0})\geq d_G(i,j_0)-2$. The claim can be proved as follows: When we run the basis pursuit over the sequence $\{\mathbf{a}_s\}_s$ to obtain \eqref{a322}, consider the $j_0$-th coordinate. For example, such a coordinate in $\mathbf{p}_{\ell}(t)$ is $\alpha_{j_0}^{\ell}(t)$. Note that $\mathbf{p}_{i_0}(t)$ with $t=d_G(i_0,m_{j_0})$ lies in the basis because the $j_0$-th element of $[\Delta_i]_{s*}$ and $[\Lambda_i]_{s*}$ for all $s$ are totally zero and the first time in the basis pursuit that the $j_0$-th coordinate is not zero occurs at $t=d_G(i_0,m_{j_0})$ in $\mathbf{p}_{i_0}(t)$. Due to this, $\mathbf{p}_{i_0}(t)$ does not lie in the span of the latest updated basis and must be added to it. Note that $t_k$ is the time when the final vector will be added to the basis, hence, $t_k \geq d_G(i_0,m_{j_0})$.
	\end{enumerate}
\end{proof}

\subsection{Lemma~\ref{lemma:case2}}
\begin{proof}
	Suppose $\mathbf{g}_{-j}$ denotes $\mathbf{g}$ after removing the $j$-th coordinate (which corresponds to $u_j$). Then, we can rewrite \eqref{ae} in the following way such that for every valid index $s$: 
	\begin{equation}\label{a27}
		\left[\mathfrak{D}^r_i\right]_{s} = \left[\mathfrak{R}_i^r\right]_{s*} \, \mathbf{g} = \left[\mathfrak{R}_i^r\right]_{sj} \,u_j + \left[\mathfrak{R}_{i,-j}^r\right]_{s*}\mathbf{g}_{-j}.
	\end{equation}
	Note that the notation $[\cdot]_{s*}$ denotes a row vector. Under Case 2 given in \eqref{case2}, there exists a row $s_0$ of $\mathfrak{R}_{i,-j}^r$ that is a linear combination of other rows of this matrix, i.e., there exist coefficients $\beta_s \in \R$ such that
	\begin{equation}\label{a28}
		\left[\mathfrak{R}_{i,-j}^r\right]_{s_0*} = \sum_{s\neq s_0} \beta_s \, \left[\mathfrak{R}_{i,-j}^r\right]_{s*}.
	\end{equation}
	Multiply both sides of \eqref{a27} by $\beta_s$ for all $s\neq s_0$, then sum over all these values and subtract the corresponding equation for $s=s_0$ from it. This leads to:
	\begin{align}\label{a29}
		\left(\sum_{s\neq s_0} \beta_s \left[\mathfrak{D}^r_i\right]_{s}\right) - \left[\mathfrak{D}^r_i\right]_{s_0} &= \underbrace{\left(\sum_{s\neq s_0} \beta_s\left[\mathfrak{R}_i^r\right]_{sj} - \left[\mathfrak{R}_i^r\right]_{s_0j}\right)}_{\neq 0} \,u_j +\underbrace{\left(\sum_{s\neq s_0} \beta_s\left[\mathfrak{R}_{i,-j}^r\right]_{sj} - \left[\mathfrak{R}_{i,-j}^r\right]_{s_0j}\right)}_{=0 \text{ from \eqref{a28}}}  \mathbf{g}_{-j}
	\end{align}
	Note that the first term is definitely not zero because otherwise Case $1$ \eqref{case1} holds rather than Case $2$ \eqref{case2}. From \eqref{a29}, $u_j$ can be obtained by the observed data points of node $i$ up to $t=t_r$ as the following:
	\begin{align}\label{a30}
	u_j = \frac{\left(\sum_{s\neq s_0} \beta_s\left[\mathfrak{R}_i^r\right]_{sj}\right) - \left[\mathfrak{R}_i^r\right]_{s_0j}}{\left(\sum_{s\neq s_0} \beta_s \left[\mathfrak{D}^r_i\right]_{s}\right) - \left[\mathfrak{D}^r_i\right]_{s_0}}.
	\end{align}
	Moreover, note that in \eqref{a30}, $\mathfrak{R}_i^r$ only depends to the underlying graph $G$ and matrix consensus $W$ both of which are known to all nodes. Hence, $u_j$ leaks before time exceeds $t$.
\end{proof}

\subsection{Theorem~\ref{lem:2}}	
For the sake of clarity and ease of computation, we need to state and prove two intermediate lemmas before proving Theorem~\ref{lem:2}. The first lemma is just a formal computing of a combination of normal random variables.
\begin{lemma}\label{integ}
	Suppose $\X$ is an arbitrary multi-dimensional non-degenerate random variable i.e., $\Sigma_{\X}$ is invertible and $A\in\R^{n\times n}$ is a constant matrix. Then $\E[\left(\X-\E[\X]\right)^{\T}A\left(\X-\E[\X]\right)] = \operatorname{tr}(A\Sigma_X)$.
\end{lemma}
\begin{proof}
	Letting $\Y= \Sigma_{\X}^{-\frac 12}\left(\X-\E[\X]\right) $, we have $\E[\Y] = \mathbf{0}$ and $\Sigma_{\Y}=\Sigma_{\X}^{-\frac 12}\Sigma_{\X}\Sigma_{\X}^{-\frac 12}=\id_n$. Hence, 
	\begin{align}
	\E[\left(\X-\E[\X]\right)^{\T}A\left(\X-\E[\X]\right)]  &= \E\left[\Y^{\T}\underbrace{\Sigma_{\X}^{\frac 12}A\Sigma_{\X}^{\frac 12}}_{B}\Y\right]\\
	& = \E\left[\sum_{i,\,j}B_{ij}Y_iY_j\right] \\
	&=   \operatorname{tr}(B) = \operatorname{tr}\left(\Sigma_{\X}^{\frac 12}A\Sigma_{\X}^{\frac 12}\right) = \operatorname{tr}\left(A\Sigma_{\X}\right).
	\end{align}
\end{proof}
Lemma~\ref{integ} eases the proof of the following lemma that computes the mutual information with specific form of arguments that will appear later in the proof of Theorem~\ref{lem:2}.
\begin{lemma}\label{lemma:1}
	Suppose $(X,Y_1,\ldots,Y_k)$ is a multivariate normal where $X\perp \Y=(Y_1,\ldots,Y_k)^{\T}$ and assume $\Sigma_{\Y}$ is invertible. Then 
	\begin{equation}
	\mathcal{I}\left(X+a_1 Y_1, \ldots, X+a_k Y_k\,;\,\, X \right) = \frac{1}{2} \log\left(1+\sigma_X^2\,\ab^{\T}\Sigma_{\Y}^{-1} \ab\right),
	\end{equation}
	where $\ab = \left[a_1,\ldots,a_k\right]^{\T}$.
\end{lemma}
\begin{proof}
	Denote $\I = \mathcal{I}\left(X+a_1 Y_1, \ldots, X+a_k Y_k\,;\,\, X \right)$ and $\Z=\Y+X\ab=\left[Z_1,\ldots,Z_k\right]^{\T}$ and note that $Z_i = X+a_iY_i$ for $i\in \setp{1,\ldots,k}$. By definition,
	\begin{equation}\label{a1p}
	\I=\int f_{X,\Z}\left(x,z_1,\ldots,z_k\right) \log \frac{f_{X, \Z}\left(x,z_1,\ldots,z_k\right)}{f_{X}(x) f_{\Z}(z_1,\ldots,z_k)} \dd x \dd z_1\ldots \dd z_k.
	\end{equation}
	To obtain the joint distribution of $(X,\Z)$, the following transformation is useful:
	\begin{equation}
	\left[\begin{array}{l}
	X \\
	\Z
	\end{array}\right]=J\left[\begin{array}{l}
	X \\
	\Y
	\end{array}\right], \quad J=\left[\begin{array}{ll}
	1 & \mathbf{0} \\
	\ab & \id_{k}
	\end{array}\right], \quad \operatorname{det}(J)=1,
	\end{equation}
	which results in
	\begin{equation}\label{a2p}
	f_{X, \Z}(x, \mathbf{z})=f_{X, \Y}(x, \mathbf{z}-\ab x)=f_{X}(x) f_{\Y}(\mathbf{z}-\ab x).
	\end{equation}
	By replacing \eqref{a2p} into \eqref{a1p} and cancelling out the term $f_X$, we get
	\begin{equation}\label{a3p}
	\I=\int f_{X}(x) f_{\Y}(\mathbf{z}-\ab x) \log \frac{f_{\Y}(\mathbf{z}-\ab x)}{f_{\Z}(\mathbf{z})} \dd x \dd \mathbf{z}.
	\end{equation}
	To compute the terms in \eqref{a3p}, we need a quick discussion. Letting $\Y \sim \mathcal{N}\left(\mu_{\Y}, \Sigma_{\Y}\right)$ and $\Z \sim \mathcal{N}\left(\mu_{\Z}, \Sigma_{\Z}\right)$, we know that 
	\begin{equation}
	\mu_{\Z} = \mu_{\Y} + \ab \mu_{X},
	\end{equation}
	and we can relate $\Sigma_{\Z}$ and $\Sigma_{\Y}$ as follows:
	\begin{equation}
	\operatorname{Cov}\left(\left[\begin{array}{c}
	X \\
	\Z
	\end{array}\right]\right)=J \operatorname{Cov}\left(\left[\begin{array}{l}
	X \\
	\Y
	\end{array}\right]\right) J^{\top}=J\left[\begin{array}{cc}
	\sigma_{X}^{2} & \mathbf{0} \\
	\mathbf{0} & \Sigma_{\Y}
	\end{array}\right] J^{\top}=\left[\begin{array}{cc}
	\sigma_{X}^{2} & * \\
	* & \Sigma_{\Y}+ \sigma_{X}^{2}\,\ab\ab^{\T}
	\end{array}\right]
	\end{equation} 
	\begin{equation}\label{sigma_rel}
	\Longrightarrow \Sigma_{\Z}=\Sigma_{\Y}+ \sigma_{X}^{2}\,\ab\ab^{\T}.
	\end{equation}
	Since $\Sigma_{\Y} \succ 0$, from \eqref{sigma_rel} we conclude that $\Sigma_{\Z} \succ 0$ and thus $\Sigma_{\Z}$ is invertible. Therefore, $\Sigma_{\Y}^{-1}$ and $\Sigma_{\Z}^{-1}$ exist. Having this, we can write
	\begin{equation}\label{a4}
	\frac{f_{\Y}(\mathbf{z}-\ab x)}{f_{\Z}(\mathbf{z})}  =\frac{\left|\Sigma_{\Z}\right|^{\frac 12}}{\left|\Sigma_{\Y}\right|^{\frac 12}} \exp \left(-\frac{1}{2}\left(\mathbf{z}-\mathbf{a} x-\mu_{\Y}\right)^{\top} \Sigma_{\Y}^{-1}\left(\mathbf{z}-\mathbf{a} x-\mu_{\Y}\right)
	+\frac{1}{2}\left(\mathbf{z}-\mu_{\Z}\right)^{\T} \Sigma_{\Z}^{-1}\left(\mathbf{z}-\mu_{\Z}\right)\right).
	\end{equation}
	Replacing \eqref{a4} into \eqref{a3p} leads to the computation of \eqref{a3p} in three terms as follows:
	\begin{equation}\label{i0}
	\I = \I_1+\I_2+\I_3,
	\end{equation}
	where the first term is
	\begin{equation}\label{i1}
	\I_1 = \int f_{X}(x) f_{\Y}(\mathbf{z}-\ab x) \log \frac{\left|\Sigma_{\Z}\right|^{\frac 12}}{\left|\Sigma_{\Y}\right|^{\frac 12}}\, \dd x \dd\z =  \frac 12\log \frac{\left|\Sigma_{\Z}\right|}{\left|\Sigma_{\Y}\right|}.
	\end{equation}
	The second term is
	\begin{align}
	\I_2 &= -\frac{1}{2}\int f_{X}(x) f_{\Y}(\mathbf{z}-\ab x) \left(\mathbf{z}-\mathbf{a} x-\mu_{\Y}\right)^{\top} \Sigma_{\Y}^{-1}\left(\mathbf{z}-\mathbf{a} x-\mu_{\Y}\right)\, \dd\z \,\dd x\\
	&= -\frac{1}{2}\int f_{X}(x) \E_{\Y}\left[\left(\Y-\mu_{\Y}\right)^{\T}\Sigma_{\Y}^{-1}\left(\Y-\mu_{\Y}\right)\right]\, \dd x \\
	&= -\frac{1}{2}\int f_{X}(x)\operatorname{tr}\left(\Sigma_{\Y}^{-1}\Sigma_{\Y}\right)\, \dd x \quad\quad \text{due to Lemma~\ref{integ}}\\
	\label{i2}&= -\frac{\operatorname{tr}\left(\id_{k}\right)}{2}\int f_{X}(x)\, \dd x = -\frac{k}{2}.
	\end{align}
	And the third term can be obtained as follows.
	\begin{align}
	\I_3 &= \frac{1}{2}\int f_{X}(x) f_{\Y}(\mathbf{z}-\ab x)\left(\mathbf{z}-\mu_{\Z}\right)^{\T} \Sigma_{\Z}^{-1}\left(\mathbf{z}-\mu_{\Z}\right)\,  \dd\z\, \dd x\\
	\label{a7}&= \frac{1}{2}\int f_{X}(x) \E_{\Y}\left[\left(\Y+\ab x-\mu_{\Z}\right)^{\T} \Sigma_{\Z}^{-1}\left(\Y+\ab x-\mu_{\Z}\right)\right]\, \dd x.
	\end{align}
	For each given $x\in\R$, we can compute the argument of the integration in \eqref{a7} as follows.
	\begin{align}
	\label{b1}\E_{\Y}\left[\left(\Y+\ab x-\mu_{\Z}\right)^{\T} \Sigma_{\Z}^{-1}\left(\Y+\ab x-\mu_{\Z}\right)\right] =\, &\E_{\Y}\left[\left(\Y-\mu_{\Y}\right)^{\T} \Sigma_{\Z}^{-1}\left(\Y-\mu_{\Y}\right)\right] +\\ \label{b2}&\E_{\Y}\left[2\left(\Y-\mu_{\Y}\right)^{\T} \Sigma_{\Z}^{-1}\left(\ab x+\mu_{\Y}-\mu_{\Z}\right)\right]+\\
	\label{b3}&\left(\ab x+\mu_{\Y}-\mu_{\Z}\right)^{\T}\Sigma_{\Z}^{-1}\left(\ab x+\mu_{\Y}-\mu_{\Z}\right).
	\end{align}
	To compute the right-hand side of the latest equation,  note that \eqref{b1} can be obtained using Lemma~\ref{integ} as follows:
	\begin{align}
	\E_{\Y}\left[\left(\Y-\mu_{\Y}\right)^{\T} \Sigma_{\Z}^{-1}\left(\Y-\mu_{\Y}\right)\right] &=\operatorname{tr}\left(\Sigma_{\Z}^{-1}\Sigma_{\Y}\right) \\
	&= \operatorname{tr}\left(\Sigma_{\Z}^{-1}\left(\Sigma_{\Z}-\sigma_{X}^{2}\,\ab\ab^{\T}\right)\right)\quad \text{due to \eqref{sigma_rel}}\\
	&=\operatorname{tr}\left(\id_k\right)-\sigma_{X}^{2}\operatorname{tr}\left(\Sigma_{\Z}^{-1}\ab\ab^{\T}\right) \\
	\label{a5}&= k-\sigma_{X}^{2}\,\ab^{\T}\Sigma_{\Z}^{-1}\ab. 
	\end{align}
	Having \eqref{a5} together with the fact that \eqref{b2} is zero and replacing $\mu_{\Y}-\mu_{\Z} = -\ab \mu_{X}$ in \eqref{b3}, we can put together the three terms and write
	\begin{align}\label{a6}
	\E_{\Y}\left[\left(\Y+\ab x-\mu_{\Z}\right)^{\T} \Sigma_{\Z}^{-1}\left(\Y+\ab x-\mu_{\Z}\right)\right] = k-\sigma_{X}^{2}\,\ab^{\T}\Sigma_{\Z}^{-1}\ab + \left(x-\mu_{X}\right)^2\ab^{\T}\Sigma_{\Z}^{-1}\ab.
	\end{align}
	Replacing \eqref{a6} in \eqref{a7} leads to
	\begin{align}
	\I_3 &= \frac{1}{2}\left(k-\sigma_{X}^{2}\,\ab^{\T}\Sigma_{\Z}^{-1}\ab\right)+\frac{1}{2}\ab^{\T}\Sigma_{\Z}^{-1}\ab\int f_{X}(x)\left(x-\mu_{X}\right)^2\dd x \\
	\label{i3}&= \frac{k}{2}-\frac{1}{2}\sigma_{X}^{2}\,\ab^{\T}\Sigma_{\Z}^{-1}\ab+\frac{1}{2}\sigma_{X}^{2}\,\ab^{\T}\Sigma_{\Z}^{-1}\ab = \frac{k}{2}.
	\end{align}
	Putting the values of $\I_1$, $\I_2$, and $\I_3$ from \eqref{i1}, \eqref{i2}, and \eqref{i3} into \eqref{i0} results in
	\begin{equation}
	\I = \I_1+\I_2+\I_3 = \frac 12\log \frac{\left|\Sigma_{\Z}\right|}{\left|\Sigma_{\Y}\right|} -\frac{k}{2}+\frac{k}{2} = \frac 12\log \frac{\left|\Sigma_{\Z}\right|}{\left|\Sigma_{\Y}\right|}.
	\end{equation}
	To simplify further, note that from \eqref{sigma_rel} and the fact that $\Sigma_{\Y}$ is invertible, we have
	\begin{equation}
	\left|\Sigma_{\Z}\right| = \left|\Sigma_{\Y}\right|\left(1+\sigma_{X}^{2}\,\ab^{\T}\Sigma_{\Y}^{-1}\ab\right).
	\end{equation}
	Therefore,
	\begin{equation}
	\I = \frac 12\log \frac{\left|\Sigma_{\Z}\right|}{\left|\Sigma_{\Y}\right|} = \frac 12\log\left(1+\sigma_{X}^{2}\,\ab^{\T}\Sigma_{\Y}^{-1}\ab\right).
	\end{equation}
\end{proof}
Finally, the proof of Theorem~\ref{lem:2} can be provided as follows using Lemma~\ref{lemma:1}. 
\begin{proof}[{\bf Proof of Theorem~\ref{lem:2}}]
	Suppose $\mathbf{g}_{-j}$ denotes $\mathbf{g}$ after removing the $j$-th coordinate (which corresponds to $u_j$). Then, we can rewrite \eqref{ae} in the following way such that for every valid index $s$: 
	\begin{equation}
	\left[\mathfrak{D}^r_i\right]_{s} = \left[\mathfrak{R}_i^r\right]_{s*} \, \mathbf{g} = \left[\mathfrak{R}_i^r\right]_{sj} \,u_j + \left[\mathfrak{R}_{i,-j}^r\right]_{s*}\mathbf{g}_{-j}.
	\end{equation}
	Note that the notation $[\cdot]_{s*}$ denotes a row vector. For ease of notation, suppose $s_{\max}$ is the number of rows of of $\mathfrak{D}^r_i$. We directly apply Lemma~\ref{lemma:1} by considering $X=u_j$, $Y_s=\left[\mathfrak{R}_{i,-j}^r\right]_{s*}\mathbf{g}_{-j}$, $\Y=[Y_1,\ldots,Y_{s_{\max}}]^{\T}$, and $a_s = \left[\mathfrak{R}_i^r\right]_{sj}$ which also gives $\ab =[a_1,\ldots,a_{s_{\max}}]^{\T}= \left[\mathfrak{R}_{i}^r\right]_{*j}$. With these replacements, the conditions of Lemma~\ref{lemma:1} hold because due to Assumption~\ref{as: 1}, $u_j$ is independent of all other randomness sources, i.e., $u_j \perp \mathbf{g}_{-j}$ which leads to $X \perp \Y$ and also we assumed in the statement of Theorem~\ref{theorem:1} that all randomnesses are Gaussian, hence, $(X,\Y)$ is multivariate normal. Further, with this assignment of $\Y$, we need to compute the covariance matrix $\operatorname{Cov}(\Y)$ and show that it is invertible. Note that
	\begin{equation}
		\Y = \mathfrak{R}_{i,-j}^r\mathbf{g}_{-j}\quad \Rightarrow \quad\Sigma_{\Y}= \operatorname{Cov}\left(\Y\right) = \mathfrak{R}_{i,-j}^r \operatorname{Cov}\left(\mathbf{g}_{-j}\right) \mathfrak{R}_{i,-j}^{r\T} = \mathfrak{R}_{i,-j}^r\mathcal{S}_{-j}\mathfrak{R}_{i,-j}^{r\T},
	\end{equation}
	where $\mathcal{S}_{-j}$ is defined in \eqref{covariance}. Note that $\mathfrak{R}_{i}^r$ is a full row-rank matrix and under Case $1$ (i.e., \eqref{case1}), $\mathfrak{R}_{i,-j}^r$ is also full row-rank. Hence, 
	\begin{equation}\label{a31}
		\operatorname{rank}\left(\Sigma_{\Y}\right) = \operatorname{rank}\left(\mathcal{S}_{-j}\right) = 2m-1.
	\end{equation}
	Since $\Sigma_{\Y}$ is a $(2m-1)\times(2m-1)$ matrix, \eqref{a31} means $\Sigma_{\Y}$ is invertible. So far, we showed that all the conditions of Lemma~\ref{lemma:1} hold under the above assignments. Therefore, the result also holds and we have
	\begin{equation}\label{a32}
	\mathcal{I}\left(\mathfrak{D}_i^r\,;\,u_j\right)=\frac{1}{2} \log\left(1+\sigma_j^2\,\ab^{\T}\Sigma_{\Y}^{-1} \ab\right),
	\end{equation}   
	Finally, due to Theorem~\ref{theorem:1}, we know that
	\begin{equation}\label{a33}
	\mathcal{I}\left(\mathfrak{D}_i(t)\,;\,u_j\right) = \mathcal{I}\left(\mathfrak{D}_i^r\,;\,u_j\right).
	\end{equation}
	With \eqref{a32} and \eqref{a33} together, we conclude that
	\begin{equation}
	\mathcal{I}\left(\mathfrak{D}_i(t)\,;\,u_j\right)=\frac{1}{2} \log\left(1+\sigma_j^2\,\ab^{\T}\Sigma_{\Y}^{-1} \ab\right).
	\end{equation} 
\end{proof}

\subsection{Theorem~\ref{theorem:2}}\label{main-proof}
\begin{proof} Suppose $G$ does not contain any generalized leaf. We want to show that Case $1$ (\cref{case1}) holds. To do so is, the main idea is that we show that under such a condition, the $j$-th column (the column corresponding to $u_j$) in $\mathfrak{R}_{i}^r$, i.e., the column to be removed to obtain $\mathfrak{R}_{i,-j}^r$ is a linear combination of other columns of $\mathfrak{R}_{i}^r$. Therefore, removing it does not change the rank. To this aim, we need to setup a notation that we only use throughout this proof. Recall from \eqref{ae} that $\mathfrak{R}_{i}^r$ is a concatenation of three matrices, $\Delta_i$, $\Lambda_i$, and $P_i^r$ as follows:
\begin{equation}\label{a3}
P_i^r =\left[\begin{array}{c}
\mathbf{p}_{j_1}^{\T}(t_{1})\\
\vdots\\
\mathbf{p}_{j_r}^{\T}(t_{r})
\end{array}  \right],\quad \mathfrak{R}_i^r = \left[\begin{array}{c}
\Delta_i\\
\Lambda_i\\
P_i^r
\end{array} \right], 
\end{equation}

Thus, each column of $\mathfrak{R}_{i}^r$ consists of three sub-columns in each of these matrices. A set of columns of $\mathfrak{R}_{i}^r$ are linearly dependent if and only if the exact same linear combination holds for the sub-columns in each of $\Delta_i$, $\Lambda_i$, and $P_i^r$.

 \noindent{\bf Notation Setup.} All the aforementioned matrices, i.e., $\mathfrak{R}_{i}^r$, $\Delta_i$, $\Lambda_i$, and $P_i^r$ have $2m$ columns. Each column corresponds to one randomness source in the model that are $u_p$ for $p\in[n]$ and $\Gamma_{ps}$ for $(p,s)\in S$. We use index $u_p$ or $\Gamma_{ps}$ to refer to those columns. For example,  $\left[\mathfrak{R}_{i}^r\right]_{*\Gamma_{ps}}$ denotes the column of $\mathfrak{R}_{i}^r$ that corresponds to $\Gamma_{ps}$. To refer to rows, we setup a notation for each sub-matrices $\Delta_i$, $\Lambda_i$, and $P_i^r$. The rows of $\Delta_i$ correspond to $u_i$ and $\Gamma_{ip}$ for $(i,p)\in S$. We also use these as row-indices. For example, $\left[\Delta_i\right]_{\Gamma_{ip}u_q}$ equals the coefficient of $u_q$ when we write $\Gamma_{ip}$ in terms of $\mathbf{g}$. This coefficient is zero since $\Gamma_{ip}$ and $u_q$ are independent noises. The rows of $\Lambda_i$ correspond to $\Gamma_{\ell i}$ for $\ell \in N_i$. Note that here $(\ell,i)$ may or may not lie in $S$. The quantity $\Gamma_{\ell i}$ is what neighbor $\ell$ sends to $i$ in the preparation phase. It may be a pure noise (when $(\ell,i)\in S$) or in the form of $u_{\ell}-\sum_{s\in N_{\ell}\setminus \setp{m_{\ell}}}\Gamma_{\ell s}$. We also use these  $\Gamma_{\ell i}$ for $\ell \in N_i$ as row-indices. Finally, for the rows of $P_i^r$, rows are characterized by pairs $(j_1,t_1),\ldots,(j_k,t_k)$, e.g., the row of $P_i^r$ corresponding to $(\ell,t)\in \{(j_1,t_1),\ldots,(j_k,t_k)\}$ is $\mathbf{p}_{\ell}^{\T}(t)$ and is denoted by $[P_i^r]_{(\ell,t) \ast}$. When we want to refer to the element $u_b$ in this row, we write $[P_i^r]_{(\ell,t) u_b}$. As we will see in the remaining of the proof, our discussion here holds for each row identically. Hence, we refer to rows by general pair $(\ell,t)$ and use this pair as a row-index. \\

Using the notation explained above, we can now state the proof. Recall that node $i$ is the attacker and node $j$ is the victim and we want to show that $\left[\mathfrak{R}_{i}^r\right]_{*u_j}$ can be generated as a linear combination of other columns of $\mathfrak{R}_{i}^r$. Note that the following arguments hold regardless of $r$ which means it holds for all $r\in\{1,\ldots,k\}$. Therefore, we omit the superscript $r$ for simplicity. Suppose the connected graph $G$ does not have any generalized leaf. Then $j$ must have a neighbor $s\neq i$. Considering all such neighbors, there exists nodes $b,s$ such that $s\in N_j$, $s\neq i$, $b\in N_s$, and $b\notin\{i,j\}$ because otherwise either the graph is not connected or we have a generalized leaf in the graph. 
%$s\notin N_i$ because otherwise we have a generalized leaf. 
%Further, we know that a degree-one node (a leaf) is also a generalized leaf. Thus, $\deg(s)\geq 2$. Hence, there exists $b \in N_i$ such that $b\neq i$. 
Now, we consider four cases based on whether $s=m_j$ versus $s\neq m_j$ and $s=m_b$ versus $s\neq m_b$. The proof in each case are provided in the following. Note that in obtaining the elements of $P_i$, we are basically using \eqref{coeff}. Again, note that the superscript $r$ will be omitted for simplicity and we use $(\ell,t)$ subscript to refer to the rows of $P_i$ where the choice of $\ell\in N_i$ and $t$ does not matter. Moreover, $\mathbf{0}$ denotes a vector with all zero elements.
\begin{enumerate}[(I)]
	\item Case $s=m_j$ and $s=m_b$:
	\begin{align}
	&\left.\begin{array}{l}
		\left[P_i\right]_{(\ell,t)u_j} = \left[W^t\right]_{\ell m_j}=\left[W^t\right]_{\ell s}\\
		\left[P_i\right]_{(\ell,t)u_b} = \left[W^t\right]_{\ell m_b} = \left[W^t\right]_{\ell s}\\
	\end{array}\right\} \Longrightarrow \quad \left[P_i\right]_{*u_j}= \left[P_i\right]_{*u_b},\\
		&\left[\Lambda_i\right]_{*u_j} = \left[\Lambda_i\right]_{*u_b} = \mathbf{0},\\
		&\left[\Delta_i\right]_{*u_j}  = \left[\Delta_i\right]_{*u_b}= \mathbf{0}.
	\end{align}
	Hence, 
	\begin{equation}\label{d4}
		\left[\mathfrak{R}_{i}\right]_{*u_j} = \left[\mathfrak{R}_{i}\right]_{*u_b}.
	\end{equation}
	\item Case $s= m_j$ and $s\neq m_b$:
	
	Note that $m_b$ might be $i$, $j$, or any other node other than $s$. Since $s\neq m_b$, we have $(b,s)\in S$ meaning that $\Gamma_{bs}$ is a pure noise. This means the relevant coefficient is computable based on $\beta_{bs}^{\ell}(t)$ in \eqref{coeff}.
	\begin{align}
	\label{c1}&\left.\begin{array}{l}
	\left[P_i\right]_{(\ell,t)u_j} = \left[W^t\right]_{\ell m_j}=\left[W^t\right]_{\ell s}\\
	\left[P_i\right]_{(\ell,t)u_b} = \left[W^t\right]_{\ell m_b} \\
	\left[P_i\right]_{(\ell,t)\Gamma_{bs}} = \left[W^t\right]_{\ell s}-\left[W^t\right]_{\ell m_b}\\
	\end{array}\right\} \Longrightarrow \quad \left[P_i\right]_{*u_j}= \left[P_i\right]_{*u_b}+\left[P_i\right]_{*\Gamma_{bs}},
	\end{align}
	Moreover,
	\begin{align}
	&\left[\Lambda_i\right]_{*u_j} = \mathbf{0},\\
	&\text{if }b\in N_i: \left\{\begin{array}{l}
		\left[\Lambda_i\right]_{bu_b} = \mathbf{1}_{m_b=i},\\
		\left[\Lambda_i\right]_{b\Gamma_{bs}} = -\mathbf{1}_{m_b=i},\\
		\text{ and other elements of }\left[\Lambda_i\right]_{*u_b} \text{ and } \left[\Lambda_i\right]_{*\Gamma_{bs}} \text{ are zero}, \\
	\end{array}\right. \\
	&\text{if }b\notin N_i: \left[\Lambda_i\right]_{*u_b}= \left[\Lambda_i\right]_{*\Gamma_{bs}} = \mathbf{0}.
	\end{align}
	Hence, we always get 
	\begin{equation}\label{c2}
	\left[\Lambda_{i}\right]_{*u_j} = \left[\Lambda_{i}\right]_{*u_b}+\left[\Lambda_{i}\right]_{*\Gamma_{bs}}.
	\end{equation}
	Further,
	\begin{equation}\label{c3}
		\left[\Delta_i\right]_{*u_j}  = \left[\Delta_i\right]_{*u_b}= \left[\Delta_i\right]_{*\Gamma_{bs}}=\mathbf{0} \quad \Longrightarrow \quad \left[\Delta_i\right]_{*u_j}  = \left[\Delta_i\right]_{*u_b}+ \left[\Delta_i\right]_{*\Gamma_{bs}}.  
	\end{equation}
	Putting \eqref{c1}, \eqref{c2}, and \eqref{c3}, we conclude that 
	\begin{equation}\label{d3}
	\left[\mathfrak{R}_{i}\right]_{*u_j} = \left[\mathfrak{R}_{i}\right]_{*u_b}+\left[\mathfrak{R}_{i}\right]_{*\Gamma_{bs}}.
	\end{equation}
	\item Case $s \neq m_j$ and $s=m_b$:
	
	Since $s\neq m_j$, we have $(j,s)\in S$ meaning that $\Gamma_{js}$ is a pure noise. This means the relevant coefficient is computable based on $\beta_{js}^{\ell}(t)$ in \eqref{coeff}.
	\begin{align}
	\label{c14}&\left.\begin{array}{l}
	\left[P_i\right]_{(\ell,t)u_j} = \left[W^t\right]_{\ell m_j}\\
	\left[P_i\right]_{(\ell,t)u_b} = \left[W^t\right]_{\ell m_b}= \left[W^t\right]_{\ell s}\\
	\left[P_i\right]_{(\ell,t)\Gamma_{js}} = \left[W^t\right]_{\ell s}-\left[W^t\right]_{\ell m_j}\\
	\end{array}\right\} \Longrightarrow \quad \left[P_i\right]_{*u_j}= \left[P_i\right]_{*u_b}-\left[P_i\right]_{*\Gamma_{js}}.
	\end{align}
	Now, for $\Lambda_i$ we have
	\begin{align}
	&\left[\Lambda_i\right]_{*u_b} = \mathbf{0},\\
	&\text{if }j\in N_i: \left\{\begin{array}{l}
	\left[\Lambda_i\right]_{ju_j} = \mathbf{1}_{m_j=i},\\
	\left[\Lambda_i\right]_{j\Gamma_{js}} = -\mathbf{1}_{m_j=i},\\
	\text{ and other elements of }\left[\Lambda_i\right]_{*u_j} \text{ and } \left[\Lambda_i\right]_{*\Gamma_{js}} \text{ are zero}, \\
	\end{array}\right. \\
	&\text{if }j\notin N_i: \left[\Lambda_i\right]_{*u_j}= \left[\Lambda_i\right]_{*\Gamma_{js}} = \mathbf{0}.
	\end{align}
	Hence, we always get 
	\begin{equation}\label{c13}
	\left[\Lambda_{i}\right]_{*u_j} = \left[\Lambda_{i}\right]_{*u_b}-\left[\Lambda_{i}\right]_{*\Gamma_{js}}.
	\end{equation}
	Further,
	\begin{equation}\label{c12}
	\left[\Delta_i\right]_{*u_j}  = \left[\Delta_i\right]_{*u_b}= \left[\Delta_i\right]_{*\Gamma_{js}}=\mathbf{0} \quad \Longrightarrow \quad \left[\Delta_i\right]_{*u_j}  = \left[\Delta_i\right]_{*u_b}- \left[\Delta_i\right]_{*\Gamma_{js}}.  
	\end{equation}
	Putting \eqref{c14}, \eqref{c13}, and \eqref{c12}, we conclude that 
	\begin{equation}\label{d2}
	\left[\mathfrak{R}_{i}\right]_{*u_j} = \left[\mathfrak{R}_{i}\right]_{*u_b}-\left[\mathfrak{R}_{i}\right]_{*\Gamma_{js}}.
	\end{equation}
	
	\item Case $s \neq m_j$ and $s\neq m_b$:
	
	First note that since $s\neq m_j$ and $s\neq m_b$, we have $(j,s),(b,s)\in S$ meaning that $\Gamma_{js}$ and $\Gamma_{bs}$ are pure noises. This means the relevant coefficients are computable based on $\beta_{js}^{\ell}(t)$ and $\beta_{bs}^{\ell}(t)$ in \eqref{coeff}.
	\begin{align}\label{c20}
	\left.\begin{array}{l}
	\left[P_i\right]_{(\ell,t)u_j} = \left[W^t\right]_{\ell m_j}\\
	\left[P_i\right]_{(\ell,t)u_b} = \left[W^t\right]_{\ell m_b}\\
	\left[P_i\right]_{(\ell,t)\Gamma_{js}} = \left[W^t\right]_{\ell s}-\left[W^t\right]_{\ell m_j}\\
	\left[P_i\right]_{(\ell,t)\Gamma_{bs}} = \left[W^t\right]_{\ell s}-\left[W^t\right]_{\ell m_b}\\
	\end{array}\right\} \Longrightarrow \quad \left[P_i\right]_{*u_j}= \left[P_i\right]_{*u_b}-\left[P_i\right]_{*\Gamma_{js}}+\left[P_i\right]_{*\Gamma_{bs}}.
	\end{align}
	Now, for $\Lambda_i$ we have
	\begin{align}
	&\text{if }j\in N_i: \left\{\begin{array}{l}
	\left[\Lambda_i\right]_{ju_j} = \mathbf{1}_{m_j=i},\\
	\left[\Lambda_i\right]_{j\Gamma_{js}} = -\mathbf{1}_{m_j=i},\\
	\text{ and other elements of }\left[\Lambda_i\right]_{*u_j} \text{ and } \left[\Lambda_i\right]_{*\Gamma_{js}} \text{ are zero}, \\
	\end{array}\right. \\
	&\text{if }j\notin N_i: \left[\Lambda_i\right]_{*u_j}= \left[\Lambda_i\right]_{*\Gamma_{js}} = \mathbf{0},\\
	&\text{if }b\in N_i: \left\{\begin{array}{l}
	\left[\Lambda_i\right]_{bu_b} = \mathbf{1}_{m_b=i},\\
	\left[\Lambda_i\right]_{b\Gamma_{bs}} = -\mathbf{1}_{m_b=i},\\
	\text{ and other elements of }\left[\Lambda_i\right]_{*u_b} \text{ and } \left[\Lambda_i\right]_{*\Gamma_{bs}} \text{ are zero}, \\
	\end{array}\right. \\
	&\text{if }b\notin N_i: \left[\Lambda_i\right]_{*u_b}= \left[\Lambda_i\right]_{*\Gamma_{bs}} = \mathbf{0}.
	\end{align}
	Hence, we always get 
	\begin{equation}\label{c21}
	\left[\Lambda_{i}\right]_{*u_j} = \left[\Lambda_{i}\right]_{*u_b}-\left[\Lambda_{i}\right]_{*\Gamma_{js}}+\left[\Lambda_{i}\right]_{*\Gamma_{bs}}.
	\end{equation}
	Further,
	\begin{equation}\label{c22}
	\left[\Delta_i\right]_{*u_j}  = \left[\Delta_i\right]_{*u_b}= \left[\Delta_i\right]_{*\Gamma_{js}}=\left[\Delta_i\right]_{*\Gamma_{bs}}=\mathbf{0} \quad \Longrightarrow \quad \left[\Delta_i\right]_{*u_j}  = \left[\Delta_i\right]_{*u_b}- \left[\Delta_i\right]_{*\Gamma_{js}}+\left[\Delta_i\right]_{*\Gamma_{bs}}.  
	\end{equation}
	Putting \eqref{c20}, \eqref{c21}, and \eqref{c22}, we conclude that 
	\begin{equation}\label{d1}
	\left[\mathfrak{R}_{i}\right]_{*u_j} = \left[\mathfrak{R}_{i}\right]_{*u_b}-\left[\mathfrak{R}_{i}\right]_{*\Gamma_{js}}+\left[\mathfrak{R}_{i}\right]_{*\Gamma_{bs}}.
	\end{equation}
\end{enumerate}
The equations \eqref{d4}, \eqref{d3}, \eqref{d2}, and \eqref{d1} show that in all cases the column $\left[\mathfrak{R}_{i}\right]_{*u_j}$ can be written as a linear combination of other columns of $\mathfrak{R}_{i}$. Hence, removing it does not change the rank, that is, for all $r\in\setp{1,\dots,k}$.
\begin{equation}
	\operatorname{rank}\left(\mathfrak{R}_{i,-j}^r\right)  = \operatorname{rank}\left(\mathfrak{R}_{i}^r\right).
\end{equation}
This means Case $1$, i.e., \eqref{case1} holds for all $r\in\setp{1,\dots,k}$.
\end{proof}

\subsection{Theorem~\ref{conv theorem}}
We start with a simple lemma that eases the computation of $t_{\epsilon}$.
\begin{lemma}\label{lem:12}
	Let $\{\mathbf{x}(t)\}_{t=0}^{\infty}$ be an arbitrary sequence of vectors over the integer $t$ that converges to $\mathbf{x}^{*}$ and suppose for every $t\geq 0$
	\begin{equation}\label{42}
	\norm{\mathbf{x}(t+1)-\mathbf{x}^{*}}_2 \leq \alpha \norm{\mathbf{x}(t)-\mathbf{x}^{*}}_2,
	\end{equation}
	for some $\alpha \in (0,1)$ and define $t_{\epsilon} = \min \setp{t \geq 0 \mid \norm{\mathbf{x}(t)-\mathbf{x}^{*}}_2 \leq \epsilon}$. Then
	\begin{equation}\label{a37}
	t_{\epsilon} \leq \ceil*{\frac{1}{\log \frac{1}{\alpha}}\cdot \log\frac{1+\norm{\mathbf{x}(0)-\mathbf{x}^{*}}_2}{\epsilon}}.
	\end{equation}
\end{lemma}
\begin{proof}
	Note that
	\begin{equation}\label{a34}
	\norm{\mathbf{x}(t)-\mathbf{x}^{*}}_2\leq \alpha^t \norm{\mathbf{x}(0)-\mathbf{x}^{*}}_2.
	\end{equation}
	When one is not considering $t\geq 0$, we can let $\tilde{t}_{\epsilon}$ be the minimum real number $t$ such that the right-hand side of \eqref{a34} holds, i.e.,
	\begin{equation}
	\tilde{t}_{\epsilon} = \min \setp{t\in \R \mid \alpha^t \norm{\mathbf{x}(0)-\mathbf{x}^{*}}_2 \leq \epsilon}.
	\end{equation}
	Then by taking the logarithm of both sides of $\alpha^t \norm{\mathbf{x}(0)-\mathbf{x}^{*}}_2 = \epsilon$, we have
	\begin{equation}\label{a36}
	\tilde{t}_{\epsilon}\log \alpha + \log\norm{\mathbf{x}(0)-\mathbf{x}^{*}}_2 = \log\epsilon \quad \Longrightarrow\quad \tilde{t}_{\epsilon} = \frac{1}{\log \frac{1}{\alpha}}\cdot \log\frac{\norm{\mathbf{x}(0)-\mathbf{x}^{*}}_2}{\epsilon}.
	\end{equation}
	Note that here $\norm{\mathbf{x}(0)-\mathbf{x}^{*}}_2$ can go to zero and cause a negative $\tilde{t}_{\epsilon}$. Hence, if we are interested in minimum non-negative time, we can define  
	\begin{equation}
	t_{\epsilon} = \min \setp{t \geq 0 \mid \norm{\mathbf{x}(t)-\mathbf{x}^{*}}_2 \leq \epsilon}.
	\end{equation}
	Having $t\geq 0$ constraint and  $\norm{\mathbf{x}(t)-\mathbf{x}^{*}}_2$ instead of $\alpha^t \norm{\mathbf{x}(0)-\mathbf{x}^{*}}_2 \leq \epsilon$ in the definition of $t_{\epsilon}$, we have $t_{\epsilon}\leq \ceil*{\tilde{t}_{\epsilon}}$ if $\tilde{t}_{\epsilon}\geq 0$ and $t_{\epsilon}=0$ otherwise. This fact together with \eqref{a36} leads to 
	\begin{equation}
	t_{\epsilon} \leq \ceil*{\frac{1}{\log \frac{1}{\alpha}}\cdot \log\frac{1+\norm{\mathbf{x}(0)-\mathbf{x}^{*}}_2}{\epsilon}}.
	\end{equation}
\end{proof}
\begin{proof}[{\bf Proof of Theorem~\ref{conv theorem}}]
	\begin{enumerate}[(i)]
		\item Define $\mathbf{u}=[u_1,\ldots,u_n]^{\T}$ and note that by construction, Algorithm~\ref{alg} preserves the summation of all values over all iterations. In other words, $\mathds{1}_n^{\T}\mathbf{v}(t) =  \mathds{1}_n^{\T}\mathbf{u} = n\mathbf{u}^*$ for all $t\geq 0$. Due to \cite{xiao2004fast}, we know that the conditions $W\mathds{1}_n = W^{\T}\mathds{1}_n = \mathds{1}_n$ and $\rho( W - {1}/{n}\mathds{1}_n\mathds{1}_n^{\T}) <1$ lead to \eqref{cons conv}. Therefore,
		\begin{equation}
			\lim_{t \to \infty } \mathbf{v}(t) = \lim_{t \to \infty } W^t\mathbf{v}(0) = \frac{1}{n}\mathds{1}_n\mathds{1}_n^{\T}\mathbf{v}(0) = \mathds{1}_n^{\T}\mathbf{u} = \mathbf{u}^*.
		\end{equation}
		\item Since Algorithm~\ref{alg} only changes $\mathbf{v}(0)$ compared to an ordinary consensus, \eqref{conv rate} trivially holds.
		\item We use Lemma~\ref{lem:12} to compute $t_{\epsilon}$ for our problem. Note that for $t\geq 0$, we have 
		\begin{align}\label{a38}
			\mathbf{v}(t+1)-\mathbf{u}^* = W\mathbf{v}(t)-W\mathbf{u}^*-\frac{1}{n}\mathds{1}_n\mathds{1}_n^{\T}\mathbf{v}(t)+\frac{1}{n}\mathds{1}_n\mathds{1}_n^{\T}\mathbf{u}^*=\left(W-\frac{1}{n}\mathds{1}_n\mathds{1}_n^{\T}\right)\left(\mathbf{v}(t)-\mathbf{u}^*\right).
		\end{align}
		In obtaining \eqref{a38}, we used the following facts:
		\begin{align}
		W\mathbf{u}^* = W\frac{1}{n}\mathds{1}_n\mathds{1}_n^{\T}\mathbf{u}^* = \frac{1}{n} W\mathds{1}_n\mathds{1}_n^{\T}\mathbf{u}^*=\frac{1}{n} \mathds{1}_n\mathds{1}_n^{\T}\mathbf{u}^* =\mathbf{u}^* \quad \text{and}\quad \mathds{1}_n^{\T}\mathbf{v}(t) =\mathds{1}_n^{\T}\mathbf{u}^* \,\text{ for all }\,\, t\geq 0.
		\end{align}
		From \eqref{a38}, we conclude that for all $t\geq 0$
		\begin{align}\label{a39}
		\norm{\mathbf{v}(t+1)-\mathbf{u}^*}_2 \leq \norm{W-\frac{1}{n}\mathds{1}_n\mathds{1}_n^{\T}}_2\norm{\mathbf{v}(t)-\mathbf{u}^*}_2.
		\end{align}
		Comparing \eqref{a39} with \eqref{42}, by putting 
		\begin{equation}\label{alpha}
			\alpha = \norm{W-\frac{1}{n}\mathds{1}_n\mathds{1}_n^{\T}}_2,
		\end{equation} 
		and noting that $\alpha<1$ (because we assumed $\rho(W-{1}/{n}\mathds{1}_n\mathds{1}_n^{\T})<1$), the quantity $t_{\epsilon}$ can be obtained from \eqref{a37} as follows:
		\begin{equation}\label{40}
		t_{\epsilon} \leq \frac{1}{\log \frac{1}{\alpha}}\cdot \log\frac{1+\norm{\mathbf{v}(0)-\mathbf{u}^{*}}_2}{\epsilon}+1\leq \frac{1}{\log \frac{1}{\alpha}}\cdot \log \left(2\cdot\frac{1+\norm{\mathbf{v}(0)-\mathbf{u}^{*}}^2_2}{\epsilon}\right)+1.
		\end{equation}
		Here, the quantity $\mathbf{v}(0)-\mathbf{u}^{*}$ is a random variable. We take the expectation of both sides of \eqref{40} and apply the Jensen inequality due to the concavity of the logarithm function (Jensen inequality states that for a concave function $f$ and random variable $X$, $\E[f(x)]\leq f(\E[X])$). Hence,
		\begin{equation}\label{41}
		\E\left[t_{\epsilon}\right] \leq \frac{1}{\log \frac{1}{\alpha}}\cdot \log\left(2\cdot\frac{1+\E\left[\norm{\mathbf{v}(0)-\mathbf{u}^{*}}^2_2\right]}{\epsilon}\right)+1.
		\end{equation}
		In the remaining of the proof, we compute $\E\left[\norm{\mathbf{v}(0)-\mathbf{u}^{*}}^2_2\right]$ and replace it into \eqref{41}. To do so, we start by writing
		\begin{align}\label{43}
			\E\left[\norm{\mathbf{v}(0)-\mathbf{u}^{*}}^2_2\right] \leq n \cdot \max_{i\in [n]}\E\left[\left(v_i(0)-u^*\right)^2\right] = n\cdot \max_{i\in [n]}\left(\E\left[v_i(0)-u^*\right]\right)^2+\operatorname{Var}\left(v_i(0)-u^*\right).
		\end{align}
		To compute the right-hand side of \eqref{43}, we denote the maximum degree of a node in $G$ by $d_{\max}$, the indicator function by $\mathbf{1}$, and the number of elements of a set $A$ by $|A|$ and continue as follows:
		\begin{align}
			\abs{\E\left[v_i(0)-u^*\right]} &= \abs{\E\left[\sum_{j\in N_i} \left(u_j-\sum_{\ell\in N_j\setminus\setp{m_j}}\Gamma_{j\ell}\right)\mathbf{1}_{i=m_j}+\Gamma_{ji}\mathbf{1}_{i\neq m_j}-\frac{1}{n}\left(u_1+\ldots+u_n\right)\right]}\\
			&\leq  \abs{\setp{j\in N_i: i=m_j}} \cdot \mu_{\max}d_{\max} + \abs{\setp{j\in N_i: i\neq m_j}}\cdot \mu_{\max} + \mu_{\max}\\
			\label{45}& \leq \left(d_{\max}^2+1\right)\mu_{\max}.
		\end{align}
		To compute $\operatorname{Var}\left(v_i(0)-u^*\right)$, we write
		\begin{align}
		\operatorname{Var}\left(v_i(0)-u^*\right) &= \operatorname{Var}\left(\sum_{j\in N_i} \left(u_j-\sum_{\ell\in N_j\setminus\setp{m_j}}\Gamma_{j\ell}\right)\mathbf{1}_{i=m_j}+\Gamma_{ji}\mathbf{1}_{i\neq m_j}-\frac{1}{n}\left(u_1+\ldots+u_n\right)\right)\\
		&\leq  \abs{\setp{j\in N_i: i=m_j}} \cdot\left( \left(1-\frac{1}{n}\right)^2\sigma_{\max}^2 + \left(d_{\max}-1\right)\sigma_{\max}^2\right)\\
		&\quad \quad+ \abs{\setp{j\in N_i: i\neq m_j}}\cdot \sigma_{\max}^2 + \abs{\setp{j\in [n]: i\neq m_j}}\cdot \frac{\sigma_{\max}^2}{n^2}\\
		&\leq   \abs{\setp{j\in N_i: i=m_j}} \cdot d_{\max}\sigma_{\max}^2 + \abs{\setp{j\in [n]: i\neq m_j}}\cdot \frac{\sigma_{\max}^2}{n^2}\\
		&\leq   d_{\max}^2\sigma_{\max}^2 + \left(n-d_{\max}\right)\frac{\sigma_{\max}^2}{n^2}\\
		\label{44}&\leq   \left(d_{\max}^2+1\right)\sigma_{\max}^2 
		\end{align}
		Putting \eqref{43}, \eqref{45}, and \eqref{44} together results in
		\begin{align}\label{46}
		\E\left[\norm{\mathbf{v}(0)-\mathbf{u}^{*}}^2_2\right] \leq  n\cdot \left(\left(d_{\max}^2+1\right)^2\mu_{\max}^2+\left(d_{\max}^2+1\right)\sigma_{\max}^2 \right) \leq n \left(d_{\max}^2+1\right)^2 \left(\mu_{\max}^2+\sigma_{\max}^2\right).
		\end{align}
		Finally, replace \eqref{alpha} and \eqref{46} into \eqref{41} to obtain
		\begin{align}
		\E\left[t_{\epsilon}\right] &\leq \frac{1}{\log \frac{1}{\norm{W-\frac{1}{n}\mathds{1}_n\mathds{1}_n^{\T}}_2}}\cdot \log\left(2\cdot\frac{1+n \left(d_{\max}^2+1\right)^2 \left(\mu_{\max}^2+\sigma_{\max}^2\right)}{\epsilon}\right)+1.
		\end{align}	
\end{enumerate}
\end{proof}

\end{document}